\documentclass[%
 reprint,
 superscriptaddress,
 floatfix,
 amsmath, amssymb,
 aps,
 pra,
longbibliography]{revtex4-2}

\usepackage{enumerate}
\usepackage{graphicx}
\usepackage{quantikz}
\usepackage{dcolumn}
\usepackage{bm}
\usepackage[english]{babel}
\usepackage[unicode, colorlinks=true,linkcolor=blue,citecolor=blue,urlcolor=blue]{hyperref}
\usepackage{amsthm}
\usepackage{amsmath}
\usepackage[dvipsnames]{xcolor}
\usepackage{dsfont}

\usepackage{xcolor}
\usepackage{ragged2e}
\usepackage[linesnumbered,ruled,vlined,noline]{algorithm2e}
\usepackage[noend]{algpseudocode}
\usepackage{booktabs}

\usepackage{changes}

\newtheorem{theorem}{Theorem}

\newtheorem{corollary}{Corollary}

\newcommand{\tr}[1]{\ensuremath{\text{Tr}\left({#1}}\right)}

\definecolor{bluepigment}{rgb}{0.2, 0.2, 0.6}
\begin{document}

\preprint{APS/123-QED}

\title{Direct channel fidelity estimation through joint fiducial grouping}

\date{\today}
\author{Júlia Barberà-Rodríguez}
\affiliation{ICFO - Institut de Ciències Fotòniques, The Barcelona Institute of Science and Technology, 08860 Castelldefels, Barcelona, Spain}

\author{Arthur Strauss}
\affiliation{Centre for Quantum Technologies, National University of Singapore, 3 Science Drive 2, Singapore 117543}

\begin{abstract}
Fault-tolerant quantum computation hinges on the requirement for low physical error rates. Reaching below threshold regime requires the accounting of circuit dependent noise, that is inherent to the execution context in which a quantum gate is usually embedded. Direct fidelity estimation is a technique that offers natural context preservation as it solely requires the insertion of local Pauli preparation and measurement fiducials around the window of interest. However, each sampled input--output Pauli pair demands its own preparation and measurement setting, an overhead that grows rapidly once the target gate is no longer Clifford. We introduce joint fiducial grouping, which partitions Pauli pairs into sets with commuting input and output operators, allowing several Pauli-transfer coefficients to be estimated within the same preparation-measurement setting. We derive an unbiased grouped estimator and finite-sample guarantees showing that grouping always reduces the number of distinct input--output settings and can also reduce the required channel uses when the target weight is concentrated within compatible groups. We characterize these gains for the parametric two-qubit gate $\mathrm{fSim}(\theta,\varphi)$, and use the grouped estimator as a context-sensitive reward for reinforcement-learning-based gate calibration. Our results provide a practical route to lower-overhead, context-preserving fidelity estimation for continuously parameterized quantum gates.
\end{abstract}

\maketitle

\section{Introduction}
Quantum processors have recently reached a level of performance at which practical quantum advantage is beginning to emerge across a growing range of applications~\cite{kim_evidence_2023,abanin_observation_2025,hartnett_fast_2026}. Continued progress now depends on mitigating increasingly subtle sources of error. As gate operations approach error rates compatible with fault-tolerant quantum computing, the remaining imperfections become harder to characterize and exhibit a pronounced dependence on the surrounding circuit context. This includes neighboring operations, idle intervals, and control crosstalk~\cite{mckay2023benchmarkingquantumprocessorperformance}. 

Characterizing these effects requires a \emph{context-preserving} fidelity estimation protocol: one that estimates the fidelity of a quantum process exactly as it appears within an application circuit, rather than after averaging over artificially modified executions. Such context-preserving fidelity estimators are also becoming increasingly important as objective functions for adaptive calibration strategies, including machine-learning-based approaches, where the same circuit must be evaluated repeatedly under realistic operating conditions.
Randomized methods~\cite{knill2008rb,arute2019xeb,erhard_characterizing_2019} typically estimate average gate or layer fidelities by randomizing the implemented circuit through twirling operations~\cite{wallman2016noise, hashim2021randomized}. Although this averaging allows scalable fidelity estimation by turning coherent errors into stochastic ones, it also removes information about how a specific circuit shapes the noise acting on the system. Other approaches recover this circuit sensitivity through coherent amplification, but often require implementing the inverse of the circuit under study~\cite{debroy2023context}, creating a bootstrapping problem when the gates being characterized are themselves imperfect.

Direct fidelity estimation (DFE)~\cite{flammia2011direct, dasilva2011practical} provides an alternative approach to estimating the fidelity of a quantum gate relative to a known target operation. It requires only local Pauli-eigenstate preparations (preparation fiducials) before the channel and local Pauli measurements (measurement fiducials) afterwards, leaving the circuit under study essentially unchanged. As a result, DFE mitigates the sampling overhead associated with full quantum process tomography~\cite{chuang1997prescription, Poyatos_1997, mohseni2008tomography} while avoiding the context-disrupting randomization inherent to randomized benchmarking protocols. These properties make DFE a particularly suitable fidelity metric for adaptive calibration schemes designed to mitigate coherent, and circuit-dependent errors.

This advantage becomes particularly relevant when the entangling gates of interest are no longer Clifford. Continuously parameterized entangling gates, such as the $\mathrm{fSim}(\theta,\varphi)$ family used on superconducting processors~\cite{acharya_quantum_2025}, are a representative example. While non-Clifford gates can be benchmarked with randomized protocols, they typically require extensions based on representation-theoretic or beyond-group constructions that can treat broader gate families~\cite{onorati2019individual,chen2022randomized}. In this setting, the standard Clifford-twirl interpretation cannot be used for practical fidelity estimation, as the averaged channel can yield a sum of decay modes rather than a single exponential with the usual fidelity conversion~\cite{dubovitskii2022partial}. Generic fractional gates hence require specific care that tomographic protocols can bypass.

Despite these advantages, standard DFE remains costly in practice. The main limitation is that each sampled input-output Pauli pair generally requires a distinct state-preparation and measurement configuration, resulting in a large experimental overhead for non-Clifford processes. This configuration overhead can be prohibitive in the context of closed-loop calibration workflows, where fidelity must be evaluated repeatedly. DFE also remains sensitive to state-preparation-and-measurement (SPAM) errors. However, readout-error mitigation (REM) can substantially reduce readout-induced bias, while state-of-the-art single-qubit gate fidelities now reach $99.99\%$ in superconducting transmons~\cite{li2023error} and $99.999\%$ in silicon spin qubits~\cite{takeda2026single}. As these local errors become smaller, context-sensitive characterization methods become increasingly practical and can complement SPAM-robust benchmarking tools by retaining information about the specific circuit and its error mechanisms~\cite{debroy2023context}.

In this work, we extend the analytical framework recently introduced in Ref.~\cite{Barber_Rodr_guez_2025} to reduce the experimental overhead needed for direct channel fidelity estimation. We introduce \emph{joint fiducial grouping} for DFE, which partitions the support of the target Pauli transfer matrix into groups whose input and output Paulis are simultaneously qubit-wise commuting (QWC). Each commuting group can be estimated with a single preparation-measurement basis pair, substantially reducing the experimental overhead of DFE. This adapts measurement-grouping techniques developed for observable estimation in variational algorithms to channel certification~\cite{kandala2017hardware, gokhale2020n,verteletskyi2020measurement, yen2020measuring, hamamura2020efficient, Crawford_2021, yen2023deterministic}. We derive finite-sample guarantees for an unbiased estimator of the channel fidelity and show that the number of channel uses is determined by the Rényi-$\tfrac12$ effective support of compatible groups. This provides an additional point of view to the known connection between the cost of direct fidelity estimation and entropic measures of nonstabilizerness of the target system~\cite{leone2023nonstabilizerness}. Joint grouping therefore offers two complementary advantages: it reduces the number of distinct preparation-measurement configurations required to characterize a process and, whenever the target Pauli transfer matrix has non-uniform weight within compatible groups, also reduces the total number of channel evaluations required to achieve a given estimation accuracy.

We validate the analytical predictions and evaluate the practical
performance of the proposed grouping strategy through numerical simulations
on the continuously parameterized $\mathrm{fSim}(\theta,\varphi)$ gate
family. We first map the predicted setting compression and shot advantage
across the full two-angle landscape, identifying where favorable grouping
exists. At a representative point, and at common requested precision and
confidence parameters, fixed-point simulations then compare the resulting
standard and grouped DFE allocations under ideal and mitigated readout.
Moreover, we probe DFE as a context-aware metric for coherent
amplification strategies by substituting the estimator of context-aware
fidelity estimation (CAFE)~\cite{debroy2023context} with its DFE
counterpart and analyze the
relevance of coupling the two frameworks. Finally, we employ the grouped
estimator as the reward function of a reinforcement-learning-based
calibration workflow. By reducing both experimental configuration overhead
and, when favorable grouping exists, the required channel budget, joint
grouped DFE provides a practical context-preserving fidelity objective for
iterative machine-learning-driven suppression of coherent gate errors.

This paper is organized as follows. In Section~\ref{sec:channel_dfe}, we review the standard DFE protocol. We introduce the joint-grouping protocol in Section~\ref{sec:joint_grouping} and derive its sample-complexity bounds in Section~\ref{sec:grouped_sample_complexity}. In Section~\ref{sec:numerics}, we validate these predictions numerically for fractional gates, considering both fidelity estimation and gate calibration under ideal and realistic conditions. Finally, we conclude in Section~\ref{sec:conclusion}.

\section{Direct fidelity estimation for channels}
\label{sec:channel_dfe}

Let $\mathcal{E}$ be an unknown $n$-qubit quantum channel that we aim to characterize and let $\mathcal{U}$ be the desired channel corresponding
to some unitary evolution. Our goal is to quantify how accurately $\mathcal{E}$ realizes $\mathcal{U}$. We do so using the entanglement fidelity, which provides a channel-level measure of agreement between the implemented and target evolutions.

We denote $d=2^n$ and use the phase-free Pauli basis
\begin{equation}
    \mathcal P_n=\{P_\alpha\}_{\alpha=1}^{d^2},
    \qquad
    \operatorname{Tr}(P_\alpha P_{\alpha'})=d\,\delta_{\alpha,\alpha'} .
\end{equation}
For any channel $\mathcal A$, define the Pauli transfer coefficients
\begin{equation}
    \chi_{\mathcal A}(\alpha,\beta)
    :=\frac{1}{d}\operatorname{Tr}\left[P_\alpha\,\mathcal A(P_\beta)\right].
    \label{eq:chi_definition}
\end{equation}
Here, $P_\beta$ is the input Pauli and $P_\alpha$ is the output Pauli observable. To represent the channel itself, we use the Liouville-space notation, in which $|P_\alpha)$ denotes the operator $P_\alpha$ viewed as a vector in operator space, with $(P_\alpha|P_\beta) = \tr{P_\alpha P_\beta}$. The channel can then be expanded in terms of the Pauli transfer coefficients as
\begin{equation}
    \mathcal A
    =\frac{1}{d}\sum_{\alpha,\beta=1}^{d^2}
    \chi_{\mathcal A}(\alpha,\beta)
    |P_\alpha)(P_\beta|.
    \label{eq:channel_expansion}
\end{equation}
We represent both the target $\mathcal{U}$ and the
implemented map $\mathcal{E}$ by their Pauli transfer matrices $\chi_\mathcal{U}$ and
$\chi_{\mathcal{E}}$ of Eq.~\eqref{eq:chi_definition}, i.e.\ as superoperators acting on the space of $d\times d$ operators. The entanglement fidelity is the normalized
Hilbert--Schmidt inner product of these two superoperators,
\begin{equation}\label{eq:entanglement-fid}
    F_e(\mathcal{U},\mathcal E)
    =\frac{1}{d^2}\operatorname{Tr}(\mathcal U^\dagger\mathcal E)
    =\frac{1}{d^2}\sum_{\alpha,\beta=1}^{d^2}
    \chi_\mathcal{U}(\alpha,\beta)\chi_{\mathcal E}(\alpha,\beta).
\end{equation}
Equivalently, $F_e=\langle U|J_{\mathcal{E}}|U\rangle$ is the overlap of the Choi state $J_{\mathcal{E}}$ of $\mathcal{E}$ 
\begin{equation}
  J_{\mathcal{E}} = (\mathcal{E}\otimes\mathcal{I})\bigl(|\Phi_d\rangle\langle\Phi_d|\bigr)
  = \frac{1}{d^{2}}\sum_{\alpha,\beta}\chi_{\mathcal{E}}(\alpha,\beta)\,P_\alpha\otimes P_\beta^{T},
  \label{eq:choi_E}
\end{equation}
with the pure Choi state $|U\rangle=(U \otimes \mathbb{I})|\Phi_d\rangle$ of the target, where $|\Phi_d\rangle = \frac{1}{\sqrt{d}}\sum_{i=1}^d |ii\rangle$ denotes the maximally entangled state. This overlap is a genuine fidelity precisely because the target is unitary, hence pure, whereas $\mathcal{E}$ may be any CPTP map with a mixed Choi state. The prefactor $1/d^{2}$ is fixed by $\sum_{\alpha,\beta}\chi_\mathcal{U}(\alpha,\beta)^{2}=d^{2}$, because unitary conjugation preserves the Hilbert--Schmidt inner product in operator space, so that
$F_e(\mathcal{U},\mathcal{U})=1$. In contrast, the Choi state of a generic CPTP map $\mathcal{E}$ is mixed, so its purity obeys
$\mathrm{Tr}(J_{\mathcal{E}}^{2})\le 1$, equivalently
$\sum_{\alpha,\beta}\chi_{\mathcal{E}}(\alpha,\beta)^{2}\le d^{2}$.

The idea of the protocol proposed in Ref.~\cite{flammia2011direct, dasilva2011practical} is to estimate the entanglement fidelity defined in Eq.~\eqref{eq:entanglement-fid} by sampling pairs of input-output Paulis, according to some probability distribution.
Since only Pauli pairs with nonzero target coefficient enter the importance-sampling distribution, we define the support
\begin{equation}
    \Omega:=\{(\alpha,\beta):\chi_\mathcal{U}(\alpha,\beta)\neq0\} .
    \label{eq:omega_support}
\end{equation}
Then the sum in Eq.~\eqref{eq:entanglement-fid} runs over the indices $(\alpha, \beta) \in \Omega$ and can be rewritten as 
\begin{equation}
    F_e(\mathcal{U},\mathcal E)
    =\sum_{(\alpha,\beta)\in \Omega}\frac{\chi_\mathcal{U}(\alpha,\beta)^2}{d^2}
    \frac{\chi_{\mathcal E}(\alpha,\beta)}{\chi_\mathcal{U}(\alpha,\beta)}.
\end{equation}
Therefore, the probability to sample an individual Pauli pair $(\alpha,\beta)\in\Omega$ is given by 
\begin{equation}
    p_{\alpha,\beta}
    =\frac{\chi_\mathcal{U}(\alpha,\beta)^2}{d^2},
    \label{eq:standard_pair_probability}
\end{equation}
and we can associate to each sampled pair the random variable
\begin{equation}
    X_{\alpha,\beta}
    =\frac{\chi_{\mathcal E}(\alpha,\beta)}{\chi_\mathcal{U}(\alpha,\beta)},
    \label{eq:standard_pair_random_variable}
\end{equation}
which defines an unbiased estimator for the entanglement fidelity, since $\mathbb E_{(\alpha,\beta)\sim p}\left[X_{\alpha,\beta}\right]=\tr{\mathcal{U}^\dagger\mathcal{E}}/d^2.$

We now describe how the entanglement fidelity can be estimated experimentally, with total additive error $2\varepsilon$ and failure probability at most $2\delta$ using Monte Carlo importance sampling. We sample 
$\ell=\lceil 1/(\varepsilon^2\delta)\rceil$ independent Pauli pairs
\[
    (\alpha_t,\beta_t)\sim p_{\alpha,\beta},
    \qquad t=1,\ldots,\ell .
\]
For a fixed sampled pair $(\alpha_t,\beta_t)$, the unknown coefficient 
$\chi_{\mathcal E}(\alpha_t,\beta_t)$ is estimated by sampling eigenstates of the input Pauli $P_{\beta_t}$. Let the input and output Paulis have the spectral decompositions
\begin{equation}
    P_{\beta_t}
    =
    \sum_{b=1}^{d}
    \lambda_t^{(b)}
    |\phi_t^{(b)}\rangle
    \langle\phi_t^{(b)}|,
    \qquad
    \lambda_t^{(b)}\in\{\pm1\},
\end{equation}
and
\begin{equation}
    P_{\alpha_t}
    =
    \sum_{a=1}^{d}
    \mu_t^{(a)}
    |\psi_t^{(a)}\rangle
    \langle\psi_t^{(a)}|,
    \qquad
    \mu_t^{(a)}\in\{\pm1\}.
\end{equation}
We define the number of times each sampled pair is measured as
\begin{equation}
    s_t
    =
    \left\lceil
    \frac{2}
    {\chi_{\mathcal U}(\alpha_t,\beta_t)^2\,\ell\,\varepsilon^2}
    \log\frac{2}{\delta}
    \right\rceil .
    \label{eq:standard_dfe_shots_per_pair}
\end{equation}
For each shot $j=1,\ldots,s_t$, we choose $b_{t,j}$ uniformly at random, prepare 
$|\phi_t^{(b_{t,j})}\rangle$, apply the channel $\mathcal E$, and measure $P_{\alpha_t}$ in its eigenbasis. If the measurement returns the output label $a_{t,j}$, we define
\begin{equation}
    B_{t,j}
    :=
    \lambda_t^{(b_{t,j})}\mu_t^{(a_{t,j})}.
\end{equation}
Then
\begin{align}
    \mathbb E[B_{t,j}]
    &=
    \frac{1}{d}
    \sum_{b_{t,j}=1}^{d}
    \lambda_t^{(b_{t,j})}
    \tr{
        P_{\alpha_t}
        \mathcal E
        \left(
            |\phi_t^{(b_{t,j})}\rangle
            \langle\phi_t^{(b_{t,j})}|
        \right)
    } \\
    &=
    \frac{1}{d}
    \tr{
        P_{\alpha_t}\mathcal E(P_{\beta_t})
   }
    =
    \chi_{\mathcal E}(\alpha_t,\beta_t).
    \label{eq:single_pair_estimator}
\end{align}
Therefore, an unbiased finite-shot estimator of 
$X_{\alpha_t,\beta_t}$ is
\begin{equation}
    \widetilde X_t
    =
    \frac{1}
    {\chi_{\mathcal U}(\alpha_t,\beta_t)s_t}
    \sum_{j=1}^{s_t} B_{t,j},
    \label{eq:finite_shot_standard_dfe_pair}
\end{equation}
and the final channel DFE estimator is the empirical average
\begin{equation}
    \widetilde F_e
    =
    \frac{1}{\ell}
    \sum_{t=1}^{\ell}
    \widetilde X_t .
    \label{eq:standard_dfe_final_estimator}
\end{equation}
With these choices the total error is at most $2\varepsilon$ with probability at least $1-2\delta$.

\section{Joint input-output estimation}
\label{sec:joint_grouping}

The standard protocol treats every Pauli pair $(\alpha,\beta)$ as a different input--output setting. Here, we group Pauli pairs that can be jointly accessed through common input and measurement bases such that each sample provides an estimate for a whole group of commuting Pauli inputs and observables. In this way, within a single input--output setting, one can estimate the Pauli transfer coefficients $\chi_\mathcal{E}(\alpha,\beta)$ of all pairs in the group simultaneously.

Here, an \emph{input--output setting} denotes the common input and output Pauli bases. We call the local circuit that prepares a sampled input eigenstate a \emph{preparation fiducial}, and the local circuit that implements the output-basis measurement a \emph{measurement fiducial}. Their ordered combination specifies an executed \emph{fiducial configuration} (or \emph{fiducial pair}). Thus, fiducial terminology describes the physical boundary operations, whereas input--output setting denotes the shared basis pair counted in the configuration-compression analysis. A fiducial pair should not be confused with a Pauli pair $(\alpha,\beta)$, which indexes a Pauli-transfer coefficient.

Let
\begin{equation}
    \mathcal G=\{G_1,\ldots,G_M\}
\end{equation}
be a partition of $\Omega$. Then, each group is written as
\begin{equation}
    G_m=
    \{(\alpha_{m,1},\beta_{m,1}),\ldots,
      (\alpha_{m,r_m},\beta_{m,r_m})\},
    \label{eq:group_definition}
\end{equation}
where $r_m:=|G_m|$ denotes the size of the group.
A valid joint input-output group satisfies two compatibility conditions. For all $l,l'\in\{1,\ldots,r_m\}$,
\begin{equation}
    [P_{\beta_{m,l}},P_{\beta_{m,l'}}]=0,
    \qquad
    [P_{\alpha_{m,l}},P_{\alpha_{m,l'}}]=0 .
    \label{eq:compatibility_conditions}
\end{equation}
In the hardware implementation we use the stronger qubit-wise commuting condition, so that the preparation- and measurement-fiducial basis changes can be realized by single-qubit rotations. The notation in Eq.~\eqref{eq:compatibility_conditions} is kept general, but throughout the numerical implementation ``commuting'' can be read as ``qubit-wise commuting''.

Equivalently, the groups can be constructed from a compatibility graph, using standard heuristics developed for measurement grouping~\cite{Crawford_2021}. Each vertex is a supported Pauli pair $(\alpha,\beta)\in\Omega$. Two vertices $(\alpha,\beta)$ and $(\tilde\alpha,\tilde\beta)$ are adjacent if
\begin{equation}
    [P_\beta,P_{\tilde\beta}]=0,
    \qquad
    [P_\alpha,P_{\tilde\alpha}]=0 .
    \label{eq:compatibility_graph_edges}
\end{equation}
A valid group is a clique of this graph, and a grouping is a partition of the vertices into compatible subsets. In practice, we use a greedy clique-cover heuristic weighted by the ideal coefficients $\chi_\mathcal{U}(\alpha,\beta)^2$.

For each group $G_m$, define the target and experimental coefficient vectors
\begin{equation}
    \mathbf u_m
    :=\left(
    \chi_\mathcal{U}(\alpha_{m,1},\beta_{m,1}),\ldots,
    \chi_\mathcal{U}(\alpha_{m,r_m},\beta_{m,r_m})
    \right),
    \label{eq:u_vector}
\end{equation}
\begin{equation}
    \mathbf e_m
    :=\left(
    \chi_{\mathcal E}(\alpha_{m,1},\beta_{m,1}),\ldots,
    \chi_{\mathcal E}(\alpha_{m,r_m},\beta_{m,r_m})
    \right).
    \label{eq:e_vector}
\end{equation}
The fidelity can then be written as a sum over groups,
\begin{equation}
    F_e(\mathcal{U},\mathcal E)
    =\frac{1}{d^2}\sum_{m=1}^{M}\mathbf u_m\cdot\mathbf e_m .
    \label{eq:fe_group_sum}
\end{equation}
The grouped analogue of DFE samples commuting groups instead of individual Pauli pairs. Define
\begin{equation}
    p_m:=\frac{\|\mathbf u_m\|_2^2}{d^2},
    \label{eq:group_probability}
\end{equation}
as the new probability distribution, since 
\begin{equation}
    \sum_{m=1}^M\|\mathbf u_m\|_2^2
    =\sum_{(\alpha,
    \beta)\in\Omega}\chi_\mathcal{U}(\alpha,\beta)^2=d^2 .
\end{equation}
For each group, we define the random variable
\begin{equation}
    X_m:=\frac{\mathbf u_m\cdot\mathbf e_m}{\|\mathbf u_m\|_2^2} .
    \label{eq:group_random_variable}
\end{equation}
Then $F_e(\mathcal{U},\mathcal E)=\sum_{m=1}^{M}p_m X_m$.
Thus, if the variables $X_m$ were known exactly, an unbiased estimator would be obtained by sampling independent and identically distributed group labels $m_t\sim p_m$ for $t=1,\ldots,\ell$ and computing
\begin{equation}
    Y:=\frac{1}{\ell}\sum_{t=1}^{\ell}X_{m_t} .
    \label{eq:outer_exact_estimator}
\end{equation}
This procedure applies because the estimator is unbiased, as follows directly from Eq.~\eqref{eq:fe_group_sum}, and has bounded variance, as shown below
\begin{align}
    \mathbb E_{m\sim p}[X_m^2]
    &=\sum_{m=1}^{M}
    \frac{\|\mathbf u_m\|_2^2}{d^2}
    \left(\frac{\mathbf u_m\cdot\mathbf e_m}{\|\mathbf u_m\|_2^2}\right)^2
    \\
    &\leq
    \frac{1}{d^2}\sum_{m=1}^{M}\|\mathbf e_m\|_2^2
    \leq 1,
    \label{eq:var_bound_outter}
\end{align}
where the first inequality is Cauchy--Schwarz and the second is the physicality bound $\sum_{\alpha,\beta}\chi_{\mathcal{E}}(\alpha,\beta)^{2}\le d^{2}$ of Eq.~\eqref{eq:choi_E}. The outer sampling variance is therefore bounded by a constant, as in standard DFE, and one can apply Monte Carlo importance sampling to this probability distribution to estimate the fidelity with additive error $\varepsilon$ and failure probability $\delta$.

Therefore, the grouped estimator does not change the fidelity being estimated. It only changes the elementary sampling unit from a single Pauli pair $(\alpha,\beta)$ to a compatible block $G_m$ whose total sampling weight is $\|\mathbf u_m\|_2^2/d^2$.

\subsection{Estimating a sampled group}
\label{subsec:estimating_sampled_group}

We now describe how to estimate in a realistic setting the contribution associated with a sampled
group in order to obtain a finite-sample estimator for the channel fidelity with additive error $\varepsilon$ and failure probability $\delta$. Let
\begin{equation}
    G_m=\{(\alpha_{m,l},\beta_{m,l})\}_{l=1}^{r_m} .
\end{equation}
be a valid group. By construction, the input Paulis $\{P_{\beta_{m,l}}\}_l$ are mutually commuting, and so are the output Paulis $\{P_{\alpha_{m,l}}\}_l$. Hence, we may choose common eigenbases
\begin{equation}
    \mathcal B_m^{\rm in}
    =
    \{|\phi_m^{(b)}\rangle\}_{b=1}^{d},
    \qquad
    \mathcal B_m^{\rm out}
    =
    \{|\psi_m^{(a)}\rangle\}_{a=1}^{d}
\end{equation}
such that, for every $l=1,\ldots,r_m$,
\begin{align}
    P_{\beta_{m,l}}
    &=
    \sum_{b=1}^{d}
    \lambda_{m,l}^{(b)}
    |\phi_m^{(b)}\rangle\langle\phi_m^{(b)}|,
    \qquad
    \lambda_{m,l}^{(b)}\in\{\pm1\},
    \label{eq:input_common_basis}
    \\
    P_{\alpha_{m,l}}
    &=
    \sum_{a=1}^{d}
    \mu_{m,l}^{(a)}
    |\psi_m^{(a)}\rangle\langle\psi_m^{(a)}|,
    \qquad
    \mu_{m,l}^{(a)}\in\{\pm1\}.
    \label{eq:output_common_basis}
\end{align}
where $\lambda_{m,l}^{(b)}$ and $\mu_{m,l}^{(a)}$ are the associated eigenvalues. The two bases need not be the same, since they diagonalize different sets of
commuting Pauli operators.

A single experimental shot for group $m$ is performed as follows. First, an
input label $b_j\in\{1,\ldots,d\}$ is sampled uniformly at random and the state
$|\phi_{m}^{(b_j)}\rangle$ is prepared. The experimental channel $\mathcal E$ is
then applied to the state, and the commuting output Paulis
$\{P_{\alpha_{m,l}}\}_l$ are measured in their common eigenbasis
$\mathcal B_m^{\rm out}$. The measurement returns an output label
$a_j\in\{1,\ldots,d\}$ with conditional probability
\begin{equation}\label{eq:group_output_probability}
    p_m(a_j|b_j) = \tr{|\psi_{m}^{(a_j)}\rangle\langle\psi_{m}^{(a_j)}|\mathcal{E}\left(|\phi_{m}^{(b_j)}\rangle\langle\phi_{m}^{(b_j)}|\right)}
\end{equation}
given that group $m$ has been sampled at this step.
This single measurement outcome determines the eigenvalue
$\mu_{m,l}^{(a_j)}$ for every output Pauli in the group. Therefore, for each
pair $(\alpha_{m,l},\beta_{m,l})\in G_m$, the same shot defines the random
variable
\begin{equation}
    B_{m,l}^{(j)}
    :=\lambda_{m,l}^{(b_j)}\mu_{m,l}^{(a_j)} .
    \label{eq:B_m_l_j}
\end{equation}
Each $B_{m,l}^{(j)}$ is an unbiased estimator of the corresponding
Pauli-transfer coefficient of $\mathcal E$. Indeed,
\begin{align}
    \mathbb E\left[B_{m,l}^{(j)}\right]
    &=\frac{1}{d}\sum_{b_j = 1}^d \lambda_{m,l}^{(b_j)} \tr{P_{\alpha_{m,l}}\mathcal{E}\left(
    |\phi_{m}^{(b_j)}\rangle\langle\phi_{m}^{(b_j)}|\right)} \nonumber \\
    &= \frac{1}{d}\tr{P_{\alpha_{m,l}}\mathcal{E}\left(P_{\beta_{m,l}}\right)} = \chi_{\mathcal E}(\alpha_{m,l},\beta_{m,l}) 
    \label{eq:B_expectation}
\end{align}
Thus, one experimental shot gives simultaneous unbiased estimates of all
Pauli-transfer coefficients in the group.

We now combine the simultaneous single-shot estimates into an estimator for the
corresponding group contribution. For a single shot $j$ of group $m$, we define
\begin{equation}
    C_m^{(j)}:=\sum_{l=1}^{r_m}\chi_\mathcal{U}(\alpha_{m,l},\beta_{m,l})B_{m,l}^{(j)}.
    \label{eq:C_j_m}
\end{equation}
Using Eq.~\eqref{eq:B_expectation}, its expectation is
\begin{equation}
    \mathbb E\left[C_m^{(j)}\right]
    = \sum_{l=1}^{r_m}\chi_\mathcal{U}(\alpha_{m,l},\beta_{m,l})
    \chi_{\mathcal E}(\alpha_{m,l},\beta_{m,l}) 
    =\mathbf u_m\cdot\mathbf e_m,
    \label{eq:C_expectation}
\end{equation}
and is bounded as $|C_m^{(j)}|\leq ||\mathbf{u}_m||_1$ since $|B_{m,l}^{(j)}| = 1$. 
Thus, $C_m^{(j)}$ is a single-shot unbiased estimator of the numerator in Eq.~\eqref{eq:group_random_variable}. In the finite-shot protocol, each sampled
group $m$ is estimated using $s_m$ independent shots, leading to the estimator
\begin{equation}
    \widetilde X_m
    :=\frac{1}{\|\mathbf u_m\|_2^2s_m}
    \sum_{j=1}^{s_m}C_m^{(j)}.
    \label{eq:X_tilde_m}
\end{equation}
The estimator is unbiased, since
\begin{equation}
    \mathbb E[\widetilde X_m]= \frac{1}{\|\mathbf u_m\|_2^2s_m}
    \sum_{j=1}^{s_m}\mathbb E\left[C_m^{(j)}\right] = \frac{\mathbf u_m\cdot\mathbf e_m}{\|\mathbf u_m\|_2^2} =  X_m .
    \label{eq:X_tilde_unbiased}
\end{equation}

Therefore, by averaging these values over all $\ell$ samples $\tilde{X}_{m_t}$ we can obtain the finite-shot grouped DFE estimator
\begin{equation}
    \label{eq:unbiased-estimator}
    \tilde{Y} = \frac{1}{\ell}\sum_{t=1}^\ell \tilde{X}_{m_t} = \frac{1}{\ell}\sum_{t=1}^\ell \sum_{j=1}^{s_{m_t}} \frac{C_{m_t}^{(j)}}{\|\mathbf u_{m_t}\|_2^2s_{m_t}}.
\end{equation}

One could instead aggregate the $\ell$ samples by median-of-means to compute the finite-shot estimate. Partitioning them into $K$ disjoint batches such that $\ell = K \tilde \ell$, computing the empirical mean of $\tilde X_{m_t}$ for each batch, and reporting the median of the batch means $\tilde Y:=\operatorname{median}\!\left\{\bar Y_1,\ldots,\bar Y_K\right\}$ is the standard route to logarithmic dependence on the confidence parameter. However, median-of-means can introduce bias in the estimate. Since the single-shot contributions within each sampled group are bounded, we retain the empirical mean and control the finite-shot fluctuations directly using Hoeffding's inequality. This preserves the unbiasedness of the fidelity estimator.

\section{Performance guarantees}
\label{sec:grouped_sample_complexity}

We now bound the number of sampled groups $\ell$ and the per-group shot budget $s_{m_t}$ needed to estimate $F_e(\mathcal{U},\mathcal{E})$ to additive accuracy $2\varepsilon$ with failure probability at most $2\delta$. The central quantity controlling the target-dependent finite-shot cost of the grouped protocol is the effective group size 
\begin{equation}
    \kappa_m := \frac{\lVert \mathbf{u}_m\rVert_1^2}{\lVert \mathbf{u}_m\rVert_2^2}.
\end{equation}
The following theorem summarizes the statistical guarantee and the resulting resource complexity. 

\begin{theorem}[Performance guarantee for grouped channel DFE]
  \label{thm:complexity}
  Let $\mathcal{U}$ be an $n$-qubit unitary and $\mathcal{E}$ a CPTP channel, set
  $d=2^{n}$, and let $\mathcal{G}=\{G_1,\dots,G_M\}$ be any compatible grouping
  of the support $\Omega$. Then, for any accuracy $\varepsilon\in(0,1)$ and failure probability $\delta\in(0,1)$, choose $\ell=\lceil1/(\varepsilon^2\delta)\rceil$ sampled groups, and 
  estimate each with the shot allocation $s_{m_t}$ of Eq.~\eqref{eq:sm} below. Then the group estimator $\tilde Y$ satisfies 
  \begin{equation}\label{eq:failure-probability}
    \Pr\!\bigl[\,\lvert \widetilde Y - F_e(\mathcal{U},\mathcal{E})\rvert \ge 2\varepsilon\,\bigr]
    \le 2\delta .
  \end{equation}
  Moreover, its expected number of channel uses obeys 
  \begin{equation}
    \mathbb{E} [N_{\mathrm{ch}}]
    = \mathcal{O}\left(\frac{1}{\varepsilon^{2}\delta}+\frac{\log{(1/\delta)}}{d^{2}\varepsilon^2}\sum_{m=1}^{M}\kappa_m\right)
    \label{eq:thm_N_ch}
  \end{equation}
  with $1\leq \kappa_m \leq r_m$ and $\sum_{m=1}^M \kappa_m \leq |\Omega|$. So the target-dependent finite-shot term in the bound is never larger than its ungrouped counterpart, while using at most $M$ distinct input--output settings. 
  \end{theorem}
  
\begin{proof}
The error of the finite-shot estimator $\widetilde{Y}$ separates into an \emph{outer} contribution from sampling only $\ell$ groups and an \emph{inner} contribution from estimating each sampled group with finitely many shots,
  \begin{equation}
    \widetilde Y - F_e
    = \underbrace{\bigl(Y - F_e\bigr)}_{\text{outer}}
    + \underbrace{\bigl(\widetilde Y - Y\bigr)}_{\text{inner}} .
    \label{eq:split}
  \end{equation}
  For the outer term, the variance bound in Eq.~\eqref{eq:var_bound_outter} allows a direct application of Chebyshev's inequality, which gives
  $\Pr[|Y-F_e|\ge\varepsilon]\le 1/(\ell\varepsilon^2)$, so choosing
  $\ell = \lceil 1 /\varepsilon^2\delta\rceil$
  makes the outer error at most $\varepsilon$ with probability at least $1-\delta$. 
  
  For the inner term, we condition on the sampled labels $\{m_t\}_{t=1}^\ell$. Given that the terms entering in Eq.~\eqref{eq:unbiased-estimator} are independent and, up to the overall prefactor $1/\ell$, bounded as
  \begin{equation}\label{eq:bound-Cj}
    \left|\frac{C_{m_t}^{(j)}}{\|\mathbf u_{m_t}\|_2^2s_{m_t}}\right| \leq \frac{\lVert \mathbf{u}_{m_t}\rVert_1}{\lVert \mathbf{u}_{m_t}\rVert_2^2\,s_{m_t}},
  \end{equation}
  using the bound in Eq.~\eqref{eq:bound-Cj}, Hoeffding's inequality therefore yields
  \begin{equation}
    \Pr\left[\lvert \widetilde{Y} - Y\rvert\ge\varepsilon \right] \le 2\exp \left(-\frac{2\varepsilon^2\ell^2}{\displaystyle\sum_{t=1}^{\ell}\sum_{j=1}^{s_{m_t}}
                 \Bigl(\tfrac{2\lVert \mathbf{u}_{m_t}\rVert_1}
                            {\lVert \mathbf{u}_{m_t}\rVert_2^2 s_{m_t}}\Bigr)^{\!2}}\right).
    \label{eq:hoeff}
  \end{equation}
  We want this probability to be at most $\delta$, which is achieved by choosing the per-group budget
  \begin{equation}
    s_{m_t} =
    \left\lceil
    \frac{2\,\lVert \mathbf{u}_{m_t}\rVert_1^2}
         {\ell\,\varepsilon^2\,\lVert \mathbf{u}_{m_t}\rVert_2^4}\,
    \log\frac{2}{\delta}
    \right\rceil,
    \label{eq:sm}
  \end{equation}
  which depends on the target coefficients of the sampled group, known in advance from $\mathcal{U}$. 
  Since this holds for every realization of the labels $\{m_t\}$, it also holds unconditionally.
  Combining the two events through the union bound,
  \begin{align}
      \Pr\!\left[
          |\widetilde Y-F_e(\mathcal{U},\mathcal E)|\geq 2\varepsilon
      \right]
      &\leq
      \Pr\!\left[
          |Y-F_e(\mathcal{U},\mathcal E)|\geq \varepsilon
      \right]
      \nonumber \\
      &\quad+
      \Pr\!\left[
          |\widetilde Y-Y|\geq \varepsilon
      \right],
      \label{eq:union_bound_outer_inner}
  \end{align}
gives Eq.~\eqref{eq:failure-probability}.

Each shot uses the experimental channel once. Thus, if group $m$ is sampled
with probability $p_m$ and is estimated with $s_m$ shots, the expected
number of channel uses is $\mathbb{E}[N_{\mathrm{ch}}]=\ell\sum_{m} p_m s_m$. Using $p_m=\lVert \mathbf{u}_m\rVert_2^2/d^2$, Eq.~\eqref{eq:sm}, and $\lceil x\rceil \leq x+1$ gives
\begin{equation}
  \mathbb{E}[N_{\mathrm{ch}}]
  \leq 1 + \frac{1}{\varepsilon^2\delta}+
    \frac{2\log{(2/\delta)}}{d^2\varepsilon^2}\sum_{m=1}^{M}\kappa_m.
  \label{eq:Nch}
\end{equation}
which proves Eq.~\eqref{eq:thm_N_ch}. Finally, $\lVert \mathbf{u}_{m_t}\rVert_1 \geq \lVert \mathbf{u}_{m_t}\rVert_2$, valid for any vector, gives $\kappa_m \geq 1$, while Cauchy-Schwarz gives $\lVert \mathbf{u}_{m}\rVert_1^2 \leq r_m \lVert \mathbf{u}_{m}\rVert_2^2$, hence $\kappa_m \leq r_m$. Summing over the partition gives $\sum_m \kappa_m\le \sum_m r_m =|\Omega|$. 
\end{proof}

The leading statistical contribution, $1/(\varepsilon^2\delta)$, is unchanged by grouping, so the protocol has the same dependence on the target precision and confidence as standard DFE. The advantage appears in the finite-shot overhead, where the number of individually estimated Pauli pairs is replaced by the total effective support. Indeed, if every Pauli pair is treated separately, as in the standard protocol, this quantity reduces to the number of nonzero pairs. Therefore, the grouped overhead is never larger than the corresponding ungrouped one. It can be substantially smaller when many compatible Pauli pairs are collected under the same measurement fiducial and their target weight is concentrated on a small effective support. Thus, grouping improves the practical channel-use cost without sacrificing the statistical guarantees of DFE, while reducing the planned input--output settings from the Pauli-pair plan to at most the $M$ group settings.

\subsection{Effective group size and entropic interpretation}
 Theorem~\eqref{thm:complexity} shows that the target-dependent finite-shot overhead is governed by the effective group sizes $\kappa_m$, rather than by the nominal group sizes $r_m$. The bounds in Eq.~\eqref{eq:thm_N_ch} imply that the grouped finite-shot contribution is never worse than the ungrouped one, and the improvement is controlled by how much smaller $\kappa_m$ is than $r_m$. 
 
This has a useful entropic interpretation. Let
$q_{m,l}:=|u_{m,l}|^2/\lVert  \mathbf{u}_m\rVert_2^2$ be the target weight conditioned
on group $G_m$. Then
\begin{equation}
  \kappa_m
  = \Bigl(\textstyle\sum_{l}\sqrt{q_{m,l}}\Bigr)^{2}
  = 2^{\,S_{1/2}(q_m)},
  \label{eq:renyi}
\end{equation}
where $S_{1/2}(q_m)=2\log_2\!\sum_l\sqrt{q_{m,l}}$ is the classical R\'enyi-$1/2$
entropy. Thus each group contributes by its R\'enyi-$1/2$ effective
support, rather than its nominal size $r_m$. This translates into the total channel uses as
\begin{equation}
  \mathbb{E}[N_{\mathrm{ch}}]
  = \mathcal{O}\!\left(\frac{1}{\varepsilon^{2}\delta}+\frac{\log(1/\delta)}{\varepsilon^2 d^2}
      \sum_{m=1}^{M} 2^{\,S_{1/2}(q_m)}\right).
\end{equation}
Writing $\Delta_m:=\log_2 r_m - S_{1/2}(q_m)\ge0$, we have
$\kappa_m=r_m\,2^{-\Delta_m}$. This quantity $\Delta_m$ measures how far the distribution is from being uniform over the whole group. Therefore, a group whose weight is spread uniformly over its $r_m$ elements has $\Delta_m=0$ and $\kappa_m = r_m$, and gives no shot improvement. By contrast, a group whose weight concentrates on $k_m\ll r_m$ pairs has
$\Delta_m\simeq\log_2(r_m/k_m)$ and contributes an effective cost
$\kappa_m\simeq k_m$.

This entropic control of the sampling cost parallels the certification results of Ref.~\cite{leone2023nonstabilizerness}, where the hardness of direct fidelity estimation is governed by a stabilizer R\'enyi entropy of the target state, or of the Choi state of the target unitary. In our setting the relevant quantity is not the global nonstabilizerness of $\mathcal{U}$, but the R\'enyi-$1/2$ entropy of the target weights conditioned on each compatible group. Grouping can therefore keep the cost low even for non-Clifford targets, provided the weight within each group remains concentrated. Corollary~\ref{cor:clifford} shows the converse situation, in which nonstabilizerness vanishes but so does the shot advantage.

\subsection{Improvement over ungrouped DFE}
\label{subsec:figures_of_merit}
 
To analyze the performance offered by the grouping technique, it is useful to distinguish its impact on two distinct figures of merit.
 
\emph{Input--output settings.} Each pair in $\Omega$ requires its own input--output setting,
so $C_{\mathrm{ungrouped}}=|\Omega|$, while grouping gives
$C_{\mathrm{grouped}}=M$. The input--output setting compression
\begin{equation}
  R_C := \frac{|\Omega|}{M}\;\ge\;1
\end{equation}
is a guaranteed input--output setting advantage, independent of the
implemented channel (see Appendix~\ref{app:configuration_complexity} for a
detailed explanation). It holds regardless of the target's coefficient
distribution and is the dominant saving on control hardware whose
compilation unit is the full circuit. In an experiment, each input--output
setting is resolved into the sampled input eigenstates. Different
eigenstates require different state-preparation instructions even when the
input Pauli basis and output-measurement basis are unchanged, each one yielding its own fiducial configuration. The total number of fiducial configurations is a proxy for the practical instruction
overhead a control system would incur, but does not necessarily translate into a direct measure of wall-clock runtime for their experimental realization. A finite sampled
schedule can use fewer settings than the planned counts, and merging Pauli
pairs that happen to share the same preparation and measurement bases can
reduce them further, in either protocol. Moreover, the wall-clock time is ultimately determined by how the control stack handles those varying instructions. More information on this is available in the Appendix ~\ref{app:configuration_complexity}, where we discuss opportunities of reducing the compilaton overhead of DFE with state-of-the-art control systems.
 
\emph{Channel uses.} When each group contains one Pauli pair ($r_m=1$, $\kappa_m=1$,
$\sum_m\kappa_m=|\Omega|$) Eq.~\eqref{eq:thm_N_ch} reduces to the channel
analogue of Flammia--Liu DFE,
\begin{equation}
  \mathbb{E}[N_{\mathrm{ch}}^{\mathrm{DFE}}]
  = \mathcal{O}\!\left(\frac{1}{\varepsilon^{2}\delta}+\frac{\log{(1/\delta)}|\Omega|}{d^{2}\varepsilon^2}\right),
\end{equation}
and we recover the known interpolation between $\mathcal{O}(1/\varepsilon^2\delta)$
for Clifford targets ($|\Omega|=d^2$) and $\mathcal{O}(d^2\log(1/\delta)/\varepsilon^2)$
for generic full-support targets ($|\Omega|\to d^4$). For a general grouping
the shot improvement factor is
\begin{equation}
  I_G := \frac{|\Omega|}{\sum_{m}\kappa_m},
  \qquad
  1\le I_G\le \frac{|\Omega|}{M}=R_C,
  \label{eq:IG}
\end{equation}
where the upper bound $R_C$ is approached precisely when every group's
weight is dominated by a single coefficient ($\kappa_m\to1$).

\begin{corollary}[Clifford targets]
\label{cor:clifford}
If $\mathcal{U}$ is Clifford then $|\chi_\mathcal{U}(\alpha,\beta)|\in\{0,1\}$, every $q_m$ is
uniform, $\kappa_m=r_m$, and $\sum_m\kappa_m=|\Omega|=d^2$. Grouping then
yields $I_G=1$ $($no worst-case shot saving$)$ but the full configuration
compression is $R_C=d^2/M$.
\end{corollary}
 
The qualitative message is that configuration compression is always
available, whereas worst-case shot compression requires the target weight
to be non-uniform within compatible blocks. In other words, the improvement in Eq.~\eqref{eq:thm_N_ch} is most relevant for structured non-Clifford dynamics. For Clifford channels, the Pauli transfer matrix is a signed permutation, and standard DFE is already essentially optimal. The useful regime is instead when the gate has nonstabilizerness but its Choi-Pauli distribution remains highly non-uniform, with most of the weight concentrated on a small effective support inside the commuting groups. In that case, grouping replaces the ungrouped count $|\Omega|$, which approaches $d^4$ for generic non-Clifford targets, by $\sum_m 2^{S_{1/2}(q_m)}$, which can be much smaller when the coefficients within each group are peaked rather than flat. Thus, the largest reductions in channel-use complexity are expected when the target Pauli-transfer weight is concentrated on a few coefficients within large compatible groups, as can occur for structured non-Clifford rotations and near-Clifford circuits. By contrast, when the weight is broadly distributed across the coefficients within each group, grouping primarily reduces the number of experimental configurations.

\section{Numerical results: Fractional gates characterization}
 \label{sec:numerics}
Numerical results are reported as the average gate fidelity, obtained from the
estimated entanglement fidelity through the standard dimensional
relation
\begin{equation}
  F(\mathcal U,\mathcal E)
  =\frac{d\,F_e(\mathcal U,\mathcal E)}{d+1} + \frac{1}{d+1}.
  \label{eq:favg}
\end{equation}
\subsection{Fractional gates}
\label{subsec:fsim_grouping_landscape}

\begin{figure*}[t]
  \centering
  \includegraphics[width=0.9\textwidth]{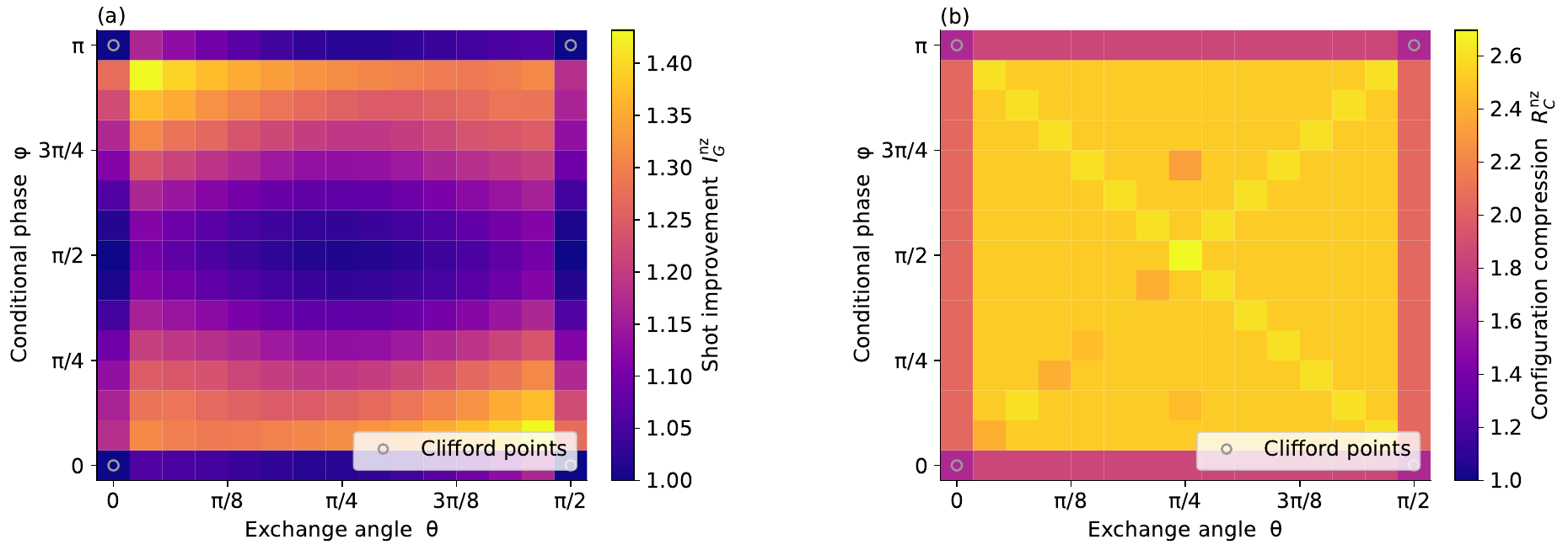}
  \caption{
  Analytic grouping expectations over the two-angle fSim family. (a) Leading weighted protocol-shot improvement
  $I_G^{\mathrm{nz}}=|\Omega_{\mathrm{nz}}|/\sum_m\kappa_m$.
  (b) Analytic input--output setting compression factor $R_C^{\mathrm{nz}}=|\Omega_{\mathrm{nz}}|/M$.
  The superscript $\mathrm{nz}$ indicates that the identity Pauli pair is omitted. The first quantity is large
  where target weight is concentrated within compatible groups, whereas the
  second records the support-to-group-count reduction independently of weight
  concentration.
  }
  \label{fig:fsim_theory_landscape}
\end{figure*}

To illustrate the advantage of joint grouping and its dependency on the target unitary, we propose a set of benchmarks over continuously parameterized two-qubit gates that are natively available on superconducting quantum hardware. Such operations, denoted recently as fractional gates~\cite{ibm_fractional_gates}, have been used in the context of utility-scale Hamiltonian-simulation workflows~\cite{acharya_quantum_2025,qedma_qesem_2026}. They are particularly suitable for reducing circuit depth as they provide more flexibility to derive efficient unitary synthesis. However, generic fractional-angle gates do not generally admit the standard Clifford-based benchmarking simplifications. This motivates direct characterization protocols that remain efficient for arbitrary input-output Pauli pairs while preserving the circuit context of the target gate.

We use the excitation-preserving family
\begin{equation}
  \mathrm{fSim}(\theta,\varphi)
  = \mathrm{CP}(-\varphi)\,R_{XX}(\theta)R_{YY}(\theta),
  \label{eq:fsim_definition}
\end{equation}
as our main numerical example. This family continuously interpolates between familiar gates including $\sqrt{\mathrm{iSWAP}} = \mathrm{fSim}(\pi/4,0)$ and iSWAP at $\mathrm{fSim}(\pi/2,0)$ up to local phases, as well as CZ $=\mathrm{fSim}(0,\pi)$ exactly, while also containing structured non-Clifford targets. This allows us to compare three regimes using the same gate family: a
near-Clifford reference, a point chosen from the grouping-advantage landscape, and
a generic non-Clifford point for which the improvement remains significant.

We first map the predicted shot advantage over the full two-angle fSim landscape.
Panel~(a) of Fig.~\ref{fig:fsim_theory_landscape} uses the same $10^{-4}$
Pauli-transfer-coefficient cutoff and greedy bilateral-QWC ordering as the
estimator. Thus the displayed ratios are computed on the truncated support;
Appendix~\ref{app:truncation_bias} gives the resulting bias bound relative to
the untruncated fidelity. The color scale shows the leading inner-shot
advantage,
\begin{equation}
  I_G^{\mathrm{nz}}
  =\frac{|\Omega_{\mathrm{nz}}|}{\sum_m\kappa_m},
  \label{eq:fsim_predicted_shot_advantage}
\end{equation}
where $\Omega_{\mathrm{nz}}:=\Omega\setminus\{(I,I)\}$ since the identity coefficient can be treated separately in post-processing, i.e. its estimate is trivial and does not require actual shots. This quantity is the ratio between the target-dependent channel-use terms of the standard and grouped protocols. Across the parameter sweep, it ranges
from $1$ to $1.43$, with the largest values forming ridges where the target
weights are concentrated within compatible groups. By contrast, the broad central band remains
close to one, showing that the gate not being Clifford alone does not imply a shot-budget
advantage.

Panel~(b) of Fig.~\ref{fig:fsim_theory_landscape} therefore reports
the theoretical compression factor
\begin{equation}
  R_C^{\mathrm{nz}}=\frac{|\Omega_{\mathrm{nz}}|}{M},
  \label{eq:fsim_predicted_config_compression}
\end{equation}
where the numerator is the number of supported non-identity Pauli pairs and
$M$ is the number of QWC groups. This factor ranges from
$1.67$ at the Clifford corners up to $2.70$ over the grid, with a broad plateau close to $2.5$. Along $\theta=0$ with $0<\varphi<\pi$, the counts are
$(|\Omega_{\mathrm{nz}}|,M)=(51,25)$, while along $\varphi=\pi$ with
$0<\theta<\pi/2$ they are $(35,19)$. Although the coefficients vary continuously along
these lines, the support and compatibility graph remain unchanged.
At the endpoints, the support changes and the compression factor returns to its
Clifford value.

\subsection{Fixed-point estimator comparison with and without readout noise}
\label{subsec:fsim_fixed_point_estimator}

The landscape prediction is tested directly at the representative point
$(\theta,\varphi)=(\pi/8,\pi/8)$. For the same fSim channel, we evaluate the grouped and standard
estimators over a range of requested precisions $\varepsilon$, both with emulated
ideal readout and with a symmetric $5\%$ readout error followed by
linear-inversion REM (see Appendix \ref{app:linear_inversion_rem} for a detailed explanation).
The failure probability is fixed to $\delta=0.1$, and the outer Pauli-sampling ceiling is
$\lceil 1/(\varepsilon^2\delta)\rceil$. The resulting protocol-shot total (one protocol shot = one use of the channel $\mathcal{E}$) is then determined by the
strategy-specific allocation over the sampled DFE configurations. Because the
same cutoff is used in these allocations, their error relative to the exact
untruncated fidelity includes any truncation bias.

\begin{figure*}[t]
  \centering
  \includegraphics[width=\textwidth]{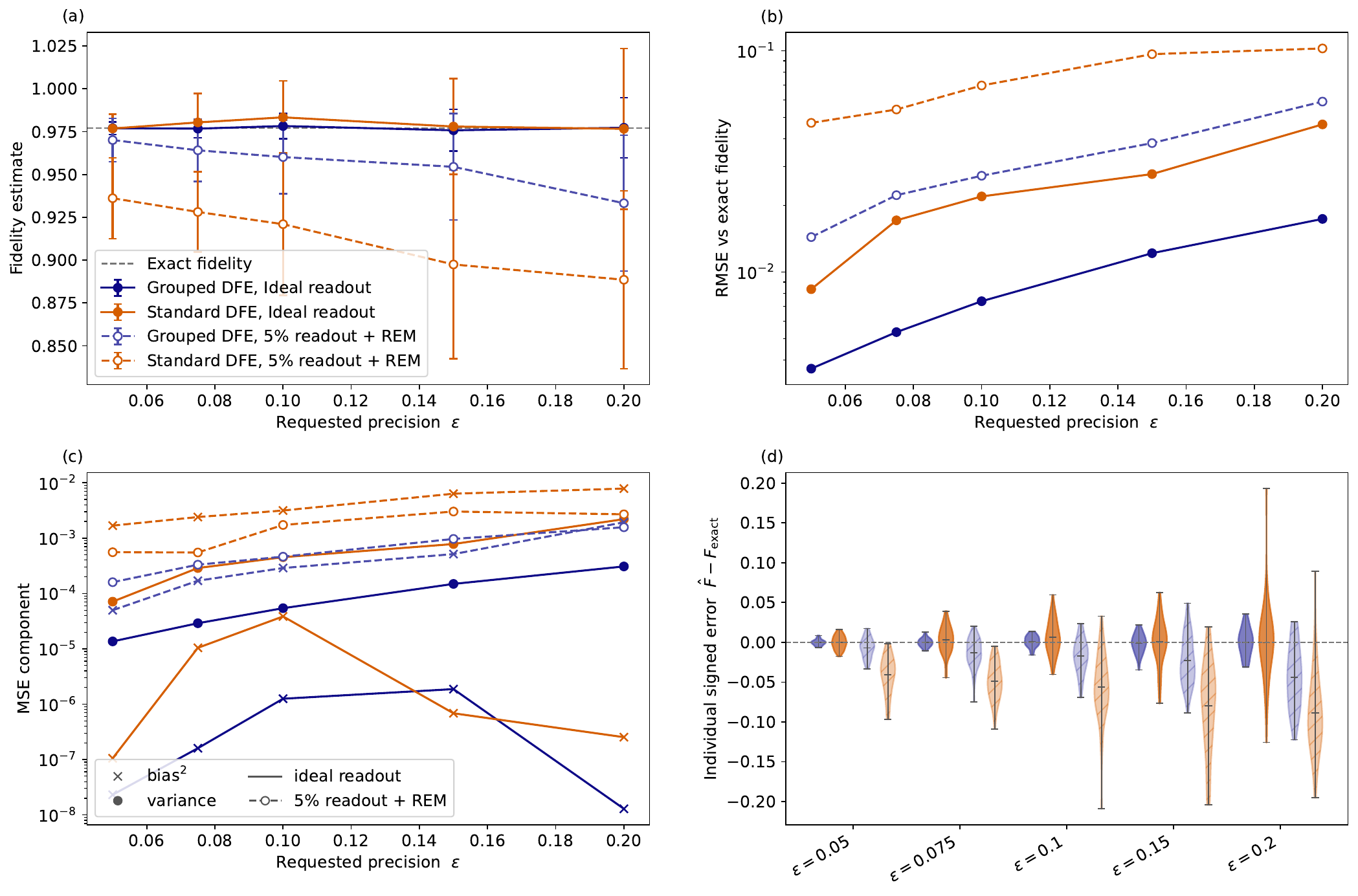}
  \caption{
  Fixed-point channel-fidelity estimation for $\mathrm{fSim}(\pi/8,\pi/8)$ under ideal readout and
  symmetric $5\%$ readout error followed by linear-inversion REM. We set $\delta=0.1$ and use 50
  independent trials at each requested precision. Indigo-blue and orange denote grouped and standard DFE,
  respectively; solid, filled marks indicate ideal readout, while dashed, open marks and light-colored
  violins indicate readout noise plus REM. (a) The trial-mean fidelity estimate with one-standard-
  deviation intervals and the exact fidelity. (b) The corresponding RMSE, which combines residual
  bias and finite-shot fluctuations. (c) Its decomposition into squared bias (crosses) and variance
  (circles). (d) The individual signed errors $\hat F-F_{\mathrm{exact}}$ from the trials summarized
  in (a). Panel (d) is the trial-level distribution underlying the mean and spread in (a): its location
  reveals residual bias, while its width, asymmetry, and tails reveal variability that panel (a)'s mean and
  one-standard-deviation interval do not resolve. For every $\varepsilon$,
  the outer sampling ceiling is $\lceil 1/(\varepsilon^2\delta)\rceil$; the displayed estimates use the
  resulting strategy-specific protocol allocations.
  }
  \label{fig:fsim_fixed_point_estimator}
\end{figure*}

The fixed-point results in Fig.~\ref{fig:fsim_fixed_point_estimator} follow the analytic prediction. Both protocols receive the same requested $(\varepsilon,\delta)$ parameters, but their theorem-derived allocations differ. Under ideal readout, the grouped estimator has lower Root Mean Squared Error (RMSE) across the full sweep, with panel~(c) showing that improvement is primarily due to a reduction in variance rather than a shift in mean. In the presence of readout noise, REM moves both estimators closer to the exact fidelity as the requested precision is increased, while the grouped estimator retains the lower RMSE over the displayed range of resulting channel-use budgets. This is an estimator-level comparison at common requested precision and confidence parameters, rather than an equal-total-execution-budget comparison.

\subsection{Inverse-free context-preserving repeated-cycle estimation}
\label{subsec:context_preserving_cycle_estimation}

Many fidelity-estimation protocols benchmark a gate in isolation. A
context-aware objective instead evaluates it in a designated execution
landscape, for example within a layer of parallel gates on neighboring
qubits. Such layer characterization can retain correlated error contributions
while preserving a local figure of merit. Randomized benchmarking can include
these contributions~\cite{mckay2023benchmarkingquantumprocessorperformance},
but reports them after averaging as part of an overall infidelity. Context-Aware
Fidelity Estimation (CAFE)~\cite{debroy2023context} takes a complementary
coherent-amplification approach, retaining the selected circuit window so that
its depth dependence can be related to the physical errors of that window.

\begin{figure*}[t]
  \centering
  \begin{minipage}[t]{0.47\textwidth}
    \centering
    \textnormal{(a) DFE}\\[-1mm]
    \[
    \begin{quantikz}[column sep=0.27cm, row sep=0.18cm]
      \lstick{$q_1$} & \gate{\ket{\psi_{\beta_1}}} &
      \gate[wires=2,style={draw=bluepigment,fill=blue!8}]{\mathcal E_{\mathrm{win}}^{\times n}} &
      \gate{R_{\alpha_1}} & \meter{} \\
      \lstick{$q_2$} & \gate{\ket{\psi_{\beta_2}}} & \qw &
      \gate{R_{\alpha_2}} & \meter{}
    \end{quantikz}
    \]
    \scriptsize Product Pauli preparation \quad $\longrightarrow$ \quad local basis rotations
  \end{minipage}\hfill
  \begin{minipage}[t]{0.47\textwidth}
    \centering
    \textnormal{(b) CAFE}\\[-1mm]
    \[
    \begin{quantikz}[column sep=0.24cm, row sep=0.32cm]
      \lstick{$q_1$} &
      \gate[wires=2,style={draw=orange,fill=orange!10}]{\substack{V_\psi,\;\ket{\psi}=V_\psi\ket{00}\\ \psi\sim\mathcal D_2}} &
      \gate[wires=2,style={draw=bluepigment,fill=blue!8}]{\mathcal E_{\mathrm{win}}^{\times n}} &
      \gate[wires=2,style={draw=red,fill=red!8}]{(\mathcal U_{\mathrm{win}}^n V_\psi)^\dagger} & \meter{} \\
      \lstick{$q_2$} & \qw & \qw & \qw & \meter{}
    \end{quantikz}
    \]
    \scriptsize Generally entangling two-design preparation \quad $\longrightarrow$ \quad generally entangling inverse
  \end{minipage}
  \caption{Circuit templates for estimating the fidelity of the same repeated
  implemented window $\mathcal E_{\mathrm{win}}^{\times n}$ (two wires shown
  schematically). (a) DFE uses local product Pauli-eigenstate preparations
  $\ket{\psi_{\beta_j}}$ and local rotations $R_{\alpha_j}$ into the output
  Pauli measurement bases. Its ideal reference determines the classical
  sampling distribution but is not applied as a circuit operation. (b) CAFE
  prepares $\ket{\psi}=V_\psi\ket{00}$ from a two-design $\mathcal D_2$ and
  appends the reference inverse $(\mathcal U_{\mathrm{win}}^nV_\psi)^\dagger$
  before the survival measurement. For a multi-qubit window, both the
  two-design preparation and the reference inverse are generally entangling, so
  their implementation can contribute error in addition to the forward window
  being assessed.}
  \label{fig:dfe_cafe_templates}
\end{figure*}
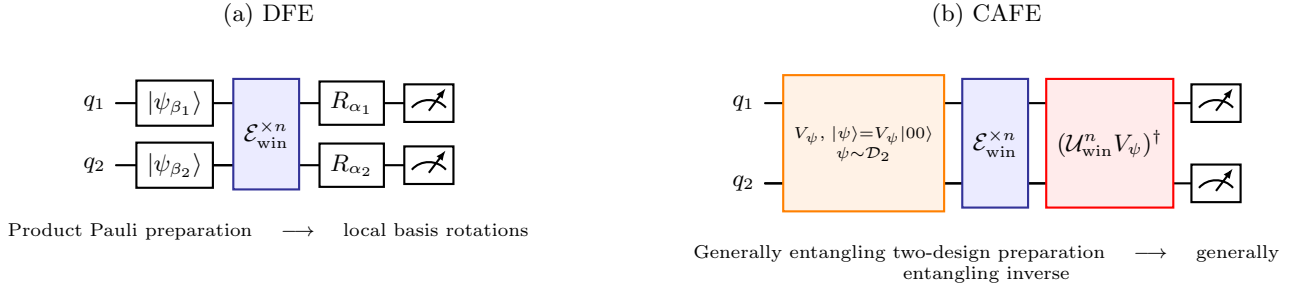

This setting is also natural for DFE, whose local Pauli-eigenstate
preparations and measurements leave the interior of the circuit under study
unchanged. We therefore compare it with the original CAFE construction as an
alternative estimator for the same context-aware repeated-cycle fidelity.
For a fixed implemented window $\mathcal E_{\mathrm{win}}$ and its ideal
reference $\mathcal U_{\mathrm{win}}$, we estimate at each repeated depth
\begin{equation}
  F^{(n)} = F\!\left(
      \mathcal E_{\mathrm{win}}^n,\mathcal U_{\mathrm{win}}^n\right).
  \label{eq:context_cycle_fidelity}
\end{equation}
Repeating the untwirled window amplifies its fixed context-dependent error
signature, from which a per-cycle infidelity can be inferred by fitting the
depth dependence. CAFE introduced a concrete way to obtain this signal:
it samples a reference state $\ket{\psi}=V_\psi\ket{0}^{\otimes n}$ from a
unitary two-design, applies the implemented window $n$ times, and appends the
reference uncompute $(\mathcal U_{\mathrm{win}}^nV_\psi)^\dagger$ before a
survival measurement~\cite{debroy2023context}. Averaging these survival
probabilities over the reference ensemble yields the depth-dependent CAFE
estimate.

DFE targets the same repeated-window fidelity through a different measurement construction. It prepares product Pauli eigenstates and measures in local Pauli bases. The ideal reference is used only to construct the classical Pauli-transfer sampling plan. No two-design state preparation and no physical inverse of $\mathcal U_{\mathrm{win}}^n$ are required. Figure~\ref{fig:dfe_cafe_templates} compares the two circuit templates.

For an already well-characterized entangling gate, the auxiliary entangling
operations in CAFE can be an acceptable part of the estimator. In a calibration
setting, however, their error can be difficult to distinguish from the error of
the forward window: the same two-qubit primitives can occur in the two-design
preparation and in the reference inverse. This is most consequential when those
primitives are precisely the operations being tuned, or are not yet known to
have sufficiently high fidelity. DFE avoids this bootstrap contribution. Its
boundary operations are local single-qubit primitives, for which very high
fidelities have been demonstrated in superconducting and semiconductor-qubit
platforms~\cite{li2023error,takeda2026single} and are traditionally much easier to obtain in an experiment. This does not make DFE
SPAM-free, but it confines these additional circuit contributions to local
operations.

The inverse-free construction is also operationally useful for continuously
parameterized windows. For CAFE, each tested parameter vector and repetition
depth requires a corresponding reference inverse, and hence an entangling
reference configuration to compile, validate, and load. For DFE, changing the
target parameters updates the classical sampling distribution while retaining
the local preparation--measurement template. This makes it natural to benchmark
a gate family or a parameterized circuit layer under fixed context, rather than
to construct a separate inverse-based experiment for every isolated gate point and number of repetitions.

For each independent shot-sampling trial and each protocol, we include the zero-depth
reference point and fit the common low-depth form used in CAFE,
\begin{align}
 \widehat{F}^{(n)}
 &= 1-\epsilon_{\mathrm{SPAM}}-n\epsilon_{\mathrm{lin}}
   -n^2\epsilon_{\mathrm{quad}} .
 \label{eq:cafe_dfe_low_depth_fit}
\end{align}
Here $\epsilon_{\mathrm{SPAM}}$ is the depth-independent error isolated by
$n=0$, while $\epsilon_{\mathrm{lin}}$ captures the leading accumulation of
incoherent error probabilities and a small systematic coherent error gives the
quadratic term $\epsilon_{\mathrm{quad}}$. This is a low-depth descriptive
model, rather than an assumption of a single decay law at arbitrary depth. We
apply the same normalization to every protocol and report its fitted one-cycle
fidelity as
\begin{equation}
 \widehat F_{1}
 = \frac{\widehat{F}^{(1)}}{\widehat{ F}^{(0)}}
 = \frac{1-\epsilon_{\mathrm{SPAM}}-\epsilon_{\mathrm{lin}}
 -\epsilon_{\mathrm{quad}}}{1-\epsilon_{\mathrm{SPAM}}}.
 \label{eq:cafe_dfe_cycle_fidelity}
\end{equation}

\begin{figure*}[t]
  \centering
  \includegraphics[width=\textwidth]{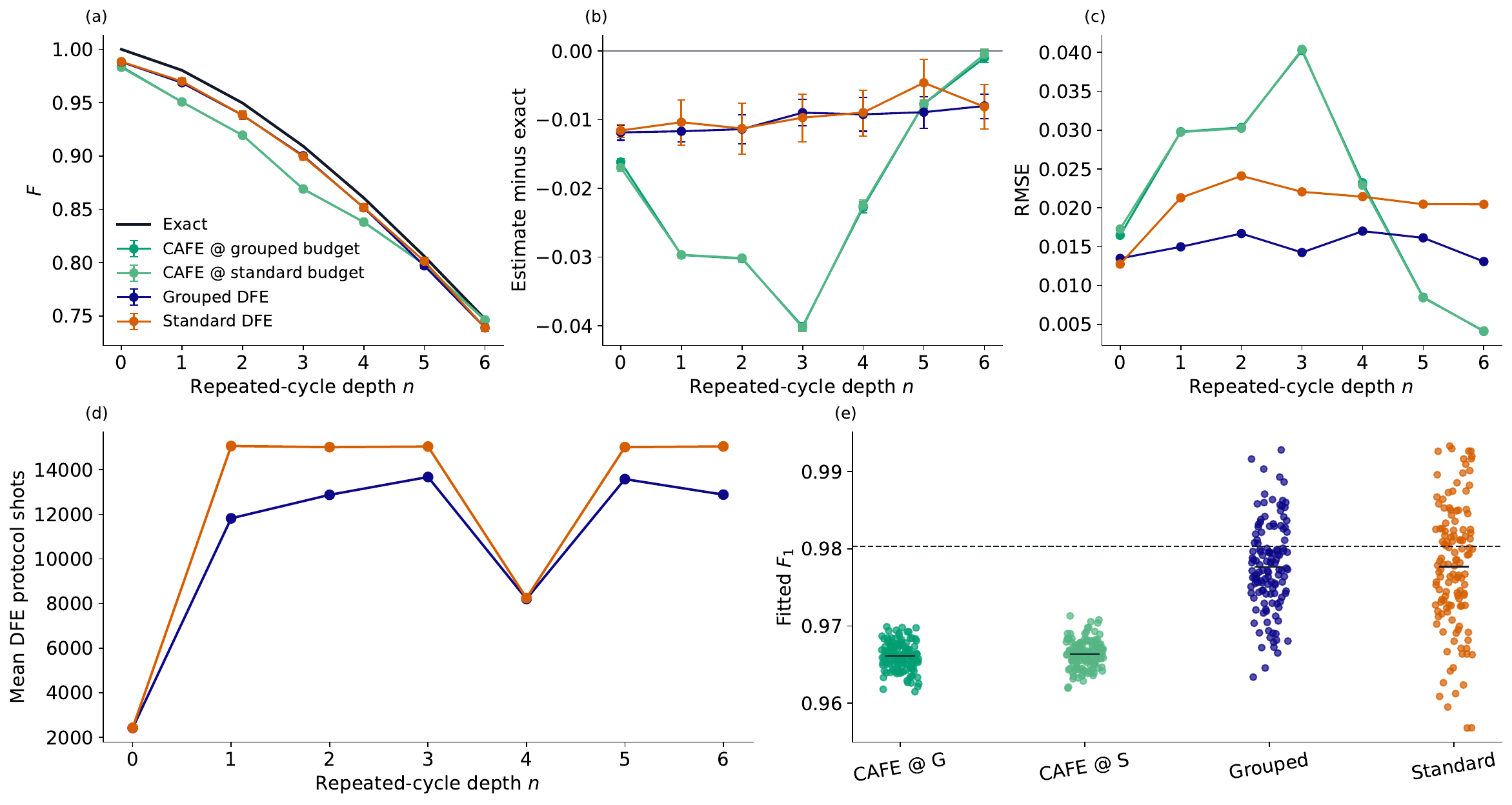}
  \caption{Fixed-context repeated-cycle estimation for
  $\mathrm{fSim}(\pi/8,\pi/8)$ at every depth $n=0,\ldots,6$, over 128
  independent shot-sampling trials. For all estimates, $(\epsilon, \delta)=(0.05, 0.1)$.
  (a) Mean depth-$n$ fidelity estimates with pointwise $95\%$
  percentile-bootstrap confidence intervals for the across-trial mean
  and the exact forward-cycle fidelity. (b) Signed estimator error and
  intervals. (c) RMSE. (d) The native grouped and standard DFE
  protocol-shot allocations; the CAFE curves are omitted because they are set
  to the realized protocol-shot budget of their
  matched DFE counterpart at every depth in each trial. (e) Cycle fidelity
  inferred by fitting each trial across the independently estimated depths
  with Eqs.~\eqref{eq:cafe_dfe_low_depth_fit} and
  \eqref{eq:cafe_dfe_cycle_fidelity}. All protocols use the common sampled
  Qiskit Aer noise model without readout-error mitigation, detailed in
  Appendix~\ref{app:simulator_noise_model}. At this selected fSim point,
  CAFE has lower trial-to-trial variability and fitted one-cycle RMSE than DFE
  at both matched budgets.}
  \label{fig:fsim_context_cafe_dfe}
\end{figure*}

Figure~\ref{fig:fsim_context_cafe_dfe} makes this comparison at the favorable
grouping point $\mathrm{fSim}(\pi/8,\pi/8)$ over depths $n=0,\ldots,6$.
At every depth, CAFE is given the same realized protocol budget as either
grouped or standard DFE, while the DFE protocols retain their native
theorem-derived allocations with $(\varepsilon,\delta)=(0.05,0.1)$. The
protocol-shot panel therefore shows only the two DFE allocations. This one comparison includes every
noise source in Table~\ref{tab:repeated_cycle_noise_model}. The plotted pointwise $95\%$
percentile-bootstrap confidence intervals quantify uncertainty in the
across-trial mean conditional on the fixed fSim target and Qiskit Aer noise
model. The selected zero-depth point identifies the estimator-specific baseline
before fitting the shared depth dependence. The mean fitted
$\epsilon_{\mathrm{SPAM}}$ is $1.59\times10^{-2}$ for CAFE at the grouped
budget, compared with $1.06\times10^{-2}$ for grouped DFE. This difference is
consistent with the two-design preparation and reference inverse contributing
their own noisy operations even at zero forward-window repetitions. After the
common normalization in Eq.~\eqref{eq:cafe_dfe_cycle_fidelity}, CAFE has lower
trial-to-trial standard deviation for the fitted one-cycle fidelity, but a
larger bias and hence a larger RMSE. At the grouped budget, the standard
deviation is $1.81\times10^{-3}$ for CAFE and $5.70\times10^{-3}$ for grouped
DFE, while the corresponding RMSE values are $1.43\times10^{-2}$ and
$6.27\times10^{-3}$. At the standard budget, the standard deviations are
$1.82\times10^{-3}$ and $8.15\times10^{-3}$, and the RMSE values are
$1.41\times10^{-2}$ and $8.54\times10^{-3}$. Thus CAFE concentrates the
fitted values more tightly, whereas grouped DFE gives the most accurate
fitted one-cycle fidelity at the matched grouped budget. These values
distinguish depth-wise precision from the precision and bias of the final
fitted cycle parameter. For
reference, fitting the exact depth-$0$--$6$ forward-fidelity curve with the
same quadratic model gives $\widehat F_1=0.97795$, below its direct
$F_1=0.98032$; panel~(e) is therefore a common descriptive fit diagnostic,
not an independent exact-cycle benchmark. The exact forward-channel fidelity
remains a reference that excludes the classical readout-assignment process, which is compensated by the fact that the noise model does not emulate readout error in this benchmark.

The gain from grouping necessarily depends on the selected channel and its
Pauli-transfer structure. At the same time, DFE draws
its boundary operations from fixed local preparation- and measurement-fiducial alphabets: varying a
continuous target parameter changes the classical sampling distribution and
shot allocation, rather than requiring a new physical two-design preparation
and reference inverse. This makes DFE a flexible route to parametric
characterization, especially on modern controllers that can update parameters
and conditional control directly within the quantum-control program
\cite{ella_quantum-classical_2023}.

Moving beyond a circuit-layer setting, where the local window is logically disentangled from its parallel context, would require a process description with substantial tomographic resources. Reference~\cite{huang2025certifying} recently showed that, for almost
all target states, local single-qubit measurements can certify the target state
and subsequently predict highly non-local properties. An important open question is whether an analogous
shadow-overlap or certification construction, applied for example to an
appropriate contextual output or Choi state, can recover local information that
supports calibration-relevant metrics for a global subcircuit. Existing shadow process-tomography approaches provide complementary channel-level starting points~\cite{levy2024classical,kunjummen2023shadow,helsen2023gateset}. 

\subsection{Calibration under ideal and readout-noisy estimator signals}
\label{subsec:calibration_protocol}
\begin{figure*}[t]
  \centering
  \includegraphics[width=\textwidth]{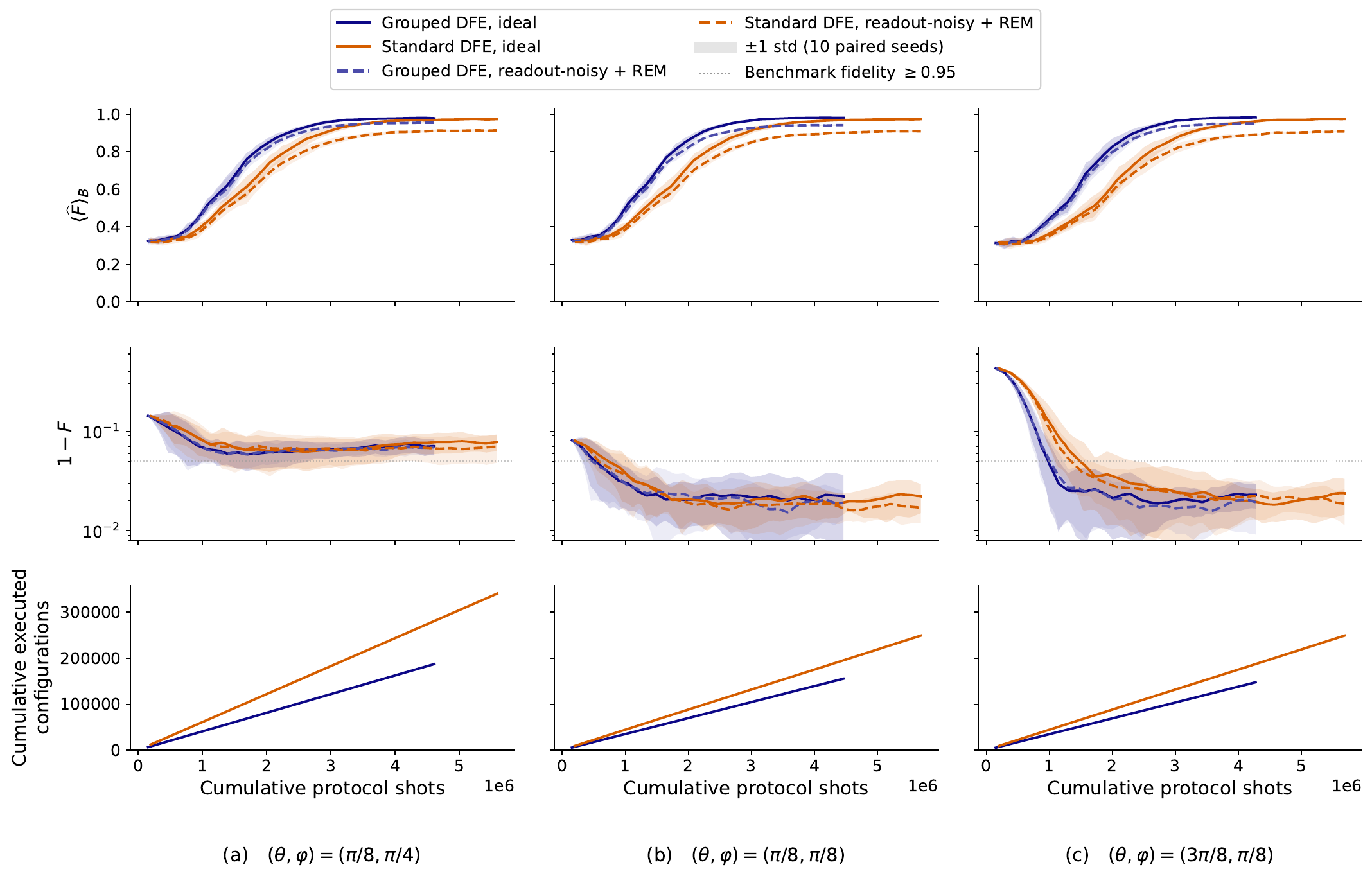}
  \caption{
  Calibration convergence for the three representative fSim targets identified at the bottom of their columns. Each column compares
  ideal readout with a symmetric $5\%$ readout error followed by linear-
  inversion REM. Solid, filled curves correspond to ideal readout, while dashed, open curves
  correspond to the readout-noisy-plus-REM condition. Blue and orange denote grouped and
  standard DFE, respectively, and shaded regions show one standard deviation
  over 10 paired seeds. The top row reports the batch-mean estimated fidelity used
  for the parameter update. The middle row gives the independently evaluated
  noiseless infidelity of the mean policy parameters. The bottom row gives
  the cumulative number of distinct execution configurations, summed over the
  action batch, as proxy for compilation and circuit-loading overhead. The configuration
  counts are identical for both readout error conditions, while the solid and dashed curves
  distinguish their estimator signals in the first two rows. Readout noise has a stronger effect on the estimated reward trajectory than on the final mean-policy infidelity, which remains comparable between the two
  conditions.
  }
  \label{fig:fsim_rl_convergence_comparison}
\end{figure*}

We finally use the two estimators as fidelity signals within a calibration loop under
two readout conditions. The five-parameter ansatz
\begin{equation}
  \mathrm{CP}(p_3)R_{XX}(p_1)R_{YY}(p_2)
  R_Z^{(0)}(p_4)R_Z^{(1)}(p_5),
  \label{eq:fsim_calibration_ansatz}
\end{equation}
can represent the target fSim gate while also allowing for local phase corrections. Following the model-free, measurement-driven control approach of Ref.~\cite{Sivak_2022}, a Proximal Policy Optimization (PPO)
agent optimizes these parameters using matched seeds, DFE precision, action
batches, and protocol-level budgets.
Details of the PPO algorithm and the training hyperparameters are provided in
Appendix~\ref{app:reinforcement_learning_calibration}.
Figure~\ref{fig:fsim_rl_convergence_comparison}
compares ideal readout with a symmetric $5\%$ readout noise followed by linear-inversion REM, using the
same cumulative protocol-shot axis. The top row shows the batch-averaged DFE
reward used by the optimizer, while the middle row reports the independently
evaluated noiseless infidelity of the mean policy parameters. These quantities can differ substantially at early stages of the training. The Gaussian policy initially has a broad
standard-deviation vector, so the $B$ actions sampled in each update probe appreciably
different circuits, whose DFE estimates are averaged into one reward signal. By contrast, the middle-row benchmark evaluates only the mean action. As the policy
contracts, the sampled actions become concentrated around the mean, and the batch-averaged reward becomes a more accurate estimate of its fidelity. Repeated evaluations in this increasingly local region also improve the statistical resolution near the optimum.

Each of the $B$ actions requires its own estimator plan, and therefore
contributes execution configurations during training. As defined in
Section~\ref{subsec:figures_of_merit}, their count is a proxy for the number of
distinct experimental instructions, rather than a direct wall-clock runtime
measurement. Such overhead is relevant because the number of circuits required
by process-fidelity optimization can itself become prohibitive~\cite{Greenaway_2021}.
The bottom row of Fig.~\ref{fig:fsim_rl_convergence_comparison} therefore shows
the cumulative number of executed configurations, summed over all actions in each batch,
for grouped and standard DFE. Grouping lowers this count by reusing compatible
bases within each action evaluation, making the compilation savings visible even when the two reward trajectories are close.

\section{Conclusion}
\label{sec:conclusion}

We introduced joint fiducial grouping for direct channel fidelity estimation, a
context-preserving way to estimate the entanglement fidelity of a target
channel while reusing compatible input preparations and output measurement
bases. The resulting estimator retains the finite-sample guarantees of DFE,
while separating two practically distinct resources: the number of channel
executions and the number of distinct input--output settings. Joint
grouping always compresses the latter, and it can also reduce the former when
the target Pauli-transfer weight is concentrated within compatible groups. Our fractional gate studies validate these predictions and show why the advantage is a
property of the target channel. The fixed-point REM comparison supplies
the same requested precision and confidence parameters to standard and grouped
DFE, then uses each protocol's resulting allocation. At those inputs, the
grouped estimator has lower observed RMSE, while one mitigated output
distribution is reused across compatible observables. Grouping does not make
assignment-matrix inversion intrinsically more accurate, and this result is
conditional on the grouping structure and shot-allocation rule.

These properties make grouped DFE a useful estimator-level primitive for
near-term utility-scale experiments. It can be inserted as a cost or reward
function in model-free calibration as well as other closed-loop pipelines that must evaluate a circuit many
times. In this role, the reduction in input--output settings can lower
compilation and controller-loading pressure, while favorable grouping can
provide more statistically informative fidelity feedback for a fixed protocol
budget. The fSim calibration experiments illustrate this intended use: the
estimator supplies a measurement-driven reward without requiring a trusted
noise model or replacing the circuit under study by a randomized surrogate.
The present evidence is simulation-based, so the remaining step toward
deployment is to quantify the same resource accounting and estimator error on
the relevant control stack and hardware. On controllers that use a FPGA for real-time control flow, preparation and measurement fiducials can be selected within an already loaded program through precompiled branches or runtime parameters. In this case, a separate circuit variant does not need to be loaded for every fiducial choice. The extent to which this reduces controller-loading overhead depends on the specific hardware and compilation policy, and remains to be quantified experimentally.

More broadly, the results emphasize the need for characterization methods that
preserve the operational context of a circuit. Coherent errors are not merely
properties of isolated gates: they can depend on neighboring operations,
idles, repetition depth, and crosstalk, and can therefore change when a
benchmarking protocol randomizes or otherwise modifies the circuit. A useful
tomographic objective for calibration should consequently resolve the process
as it is embedded in the circuit where it will be used, while remaining light
enough to serve inside an optimization loop. Grouped DFE is one construction
in this direction, combining that context sensitivity with an explicit and
auditable resource trade-off.

An open direction is the connection between deterministic grouping and shadow-based process characterization~\cite{levy2024classical,kunjummen2023shadow,helsen2023gateset}. Classical shadows of channels allow many fidelity queries to be answered from a single randomized data set, but under local measurements their sample complexity can scale less favorably with system size than standard DFE because it is controlled by the corresponding shadow norms. Derandomization results~\cite{huang2021efficient} suggest that randomized shadow protocols and deterministic grouping may be viewed as two ends of a broader family of measurement strategies. Calibration provides a natural intermediate regime: the target changes only slightly between iterations, so grouped configurations could be reused across several nearby fidelity queries. A target-biased, grouped analogue of process shadows could therefore combine the reusability of shadow data with the Rényi-$1/2$ effective-support scaling derived here. A complementary extension would be to allow overlapping groups, which could further reduce the sampling cost~\cite{wu2023overlapped}.

\begin{acknowledgments}
This project has been supported by the Government of Spain (Severo Ochoa CEX2019-000910-S and FUNQIP), Fundació Cellex, Fundació Mir-Puig, Generalitat de Catalunya (CERCA program). J.B.R. has received funding from the “Secretaria d’Universitats i Recerca del Departament de Recerca i Universitats de la Generalitat de Catalunya” under grant FI-3 00096, as well as the European Social Fund Plus.
A.S. has been supported by the National Research Foundation, Singapore through the National Quantum Office, hosted in A*STAR, under its Centre for Quantum Technologies Funding Initiative (S24Q2d0009).
J.B.R. thanks Pr.\ Antonio Acín for useful conversations. A.S.\ acknowledges useful discussions with Pr.\ Hui Khoon Ng on this project, as well as key contributions to the software stack required to run the machine-learning driven calibration experiments from Aniket Chatterjee and Lukas Voss.
The theoretical content was developed without the use of paid LLMs. The numerical results were obtained from a software codebase constructed by A.S.
Testing scripts and use cases were written with OpenAI ChatGPT~5.6 models and subsequently reviewed by the authors.
The same tool was used for minor editing of the main text; the authors have reviewed the manuscript and take full responsibility for its content.
\end{acknowledgments}

\bibliography{bibliography}
\appendix

\section{Input--output setting complexity}
\label{app:configuration_complexity}

The most direct resource reduced by joint input-output grouping is the number of distinct input--output settings. We first count the planned settings before optional deduplication of identical basis pairs. Without grouping, each supported Pauli pair requires one preparation basis and one measurement basis, so $|\Omega|$ planned settings are needed. Therefore,
\begin{equation}
    C_{\mathrm{ungrouped}}=|\Omega| .
    \label{eq:C_ungrouped}
\end{equation}
Here $C$ is the planned input--output setting count used in
Section~\ref{subsec:figures_of_merit}; it does not include the individual input
eigenstates sampled within a fixed input basis. With grouping, each group $G_m$
requires one common input basis and one common output basis, so 
\begin{equation}
    C_{\mathrm{grouped}}=M,
    \label{eq:C_grouped}
\end{equation}
where $M$ is the total number of commuting groups. A finite sampled schedule
can use a strict subset of these settings, and setting-level deduplication can
reduce the realized count further. The input--output setting compression factor
of the planned grouping $\mathcal{G}$ is
\begin{equation}
    R_C(\mathcal G)
    :=\frac{C_{\mathrm{ungrouped}}}{C_{\mathrm{grouped}}}
    =\frac{|\Omega|}{M} .
    \label{eq:configuration_compression}
\end{equation}
This is the clearest guaranteed advantage of the protocol. Even when the worst-case shot bound is unchanged, grouping can substantially reduce compilation overhead, circuit loading, and the number of distinct state-preparation and measurement settings used in a calibration loop.

Beyond shot economy, configuration compression affects the classical-control workflow. In a conventional arbitrary-waveform-generator (AWG) stack, each distinct circuit configuration is compiled into a complete multi-qubit pulse program: the AWG stores precomputed waveform samples and the associated timing/sequencing instructions. Since the global timing has already been resolved, the conventional workflow considered here loads that full program for each configuration rather than dispatching reusable per-qubit subsequences at run time. When calibrating two-qubit gates in parallel across a chip, the number of programs to compile and load consequently scales as $|\Omega| \times N_{\text{pairs}}$ without grouping. 
Grouping reduces this count to $M \times N_{\text{pairs}}$, reducing program-management and loading work; the resulting wall-clock effect is hardware-dependent and is not measured here. This burden can be amplified in randomized protocols such as randomized benchmarking, where the number of circuit instances already grows with sequence length and the number of random Clifford samples. The AWG scaling is not universal: a FPGA-based controller can execute some conditional or parameter-selection logic in real time. If the required basis changes are available as precompiled branches or runtime parameters, the controller can select them within an already loaded program rather than loading a new full-circuit variant. In that setting, configuration compression still reduces the number of distinct experiment descriptions, but need not translate into the same upload-count scaling; the outcome depends on the controller capabilities and compilation policy.

\section{Truncation and bias}
\label{app:truncation_bias}

For non-Clifford channels, the support $\Omega$ can be large, and small target coefficients contribute negligibly to the fidelity. In practice, one may retain only a subset of important Pauli pairs. Define the truncated support:
\begin{equation}\label{eq:truncation-set}
    \Omega_\tau := \left\{(\alpha, \beta) : | \chi_\mathcal{U}(\alpha, \beta)| \geq\tau\right\},
\end{equation}
for some threshold $\tau \geq 0$ (so that $\Omega_\tau \subseteq \Omega$), and let $\mathcal{G}_\tau$ be the grouping restricted to $\Omega_\tau$. Then, the truncated fidelity is
\begin{equation}
    F_{\tau}(\mathcal U,\mathcal E)
    :=\frac{1}{d^2}\sum_{(\alpha,\beta)\in\Omega_\tau}
    \chi_\mathcal{U}(\alpha,\beta)\chi_{\mathcal E}(\alpha,\beta).
    \label{eq:truncated_fidelity}
\end{equation}
The same grouped estimator applies after replacing $\Omega$ by $\Omega_\tau$, so we form the groups on the set defined in Eq.~\eqref{eq:truncation-set}. In that case, the group probabilities are normalized by the retained weight
\begin{equation}
    p_m^{(\tau)}=\frac{\|\mathbf u_m\|_2^2}{\sum_{(\alpha,\beta)\in \Omega_\tau}\chi_\mathcal{U}(\alpha,\beta)^2} .
    \label{eq:truncated_probabilities}
\end{equation}

Since these probabilities are normalized by the retained weight $W_\tau:= \sum_{(\alpha,\beta)\in\Omega_\tau}\chi_\mathcal{U}(\alpha,\beta)^2$ rather than by $d^2$, the final empirical average must be rescaled by $W_\tau/d^2$. This follows from
\begin{align}
    \mathbb{E}_m[X_m] &= \sum_m \frac{\|\mathbf u_m\|_2^2}{W_\tau}\frac{\mathbf{u}_m\cdot \mathbf{e}_m}{\|\mathbf u_m\|_2^2} = \frac{1}{W_\tau}\sum_m \mathbf{u}_m\cdot \mathbf{e}_m \nonumber \\
    & = \frac{1}{W_\tau}\sum_{(\alpha,\beta)\in\Omega_\tau} \chi_\mathcal{U}(\alpha,\beta)\chi_{\mathcal E}(\alpha,\beta).
\end{align}
Because $F_\tau$ is defined in Eq.~\eqref{eq:truncated_fidelity}, $\mathbb{E}_m[X_m] = \frac{d^2}{W_\tau}F_\tau$. After rescaling, the estimator $Y^\tau$ satisfies $\mathbb{E}[Y^\tau]=F_\tau$. With respect to the untruncated fidelity, it is biased, with 
\begin{equation}
    |F_\tau -F_e(\mathcal{U},\mathcal{E})|
    =
    \left |\frac{1}{d^2}\sum_{(\alpha,\beta)\notin \Omega_\tau}
    \chi_\mathcal{U}(\alpha,\beta) \chi_\mathcal{E}(\alpha,\beta)\right |.
    \label{eq:truncation_l1_bias}
\end{equation}
By Cauchy--Schwarz we can write
\begin{align}
     |F_\tau -F_e(\mathcal{U},\mathcal{E})|
    &\leq
    \frac{1}{d^2}
    \sqrt{
    \sum_{(\alpha,\beta)\notin\Omega_\tau}
    \chi_\mathcal{U}(\alpha,\beta)^2} \nonumber \\
    &\times \sqrt{
    \sum_{(\alpha,\beta)\notin\Omega_\tau}
    \chi_\mathcal{E}(\alpha,\beta)^2}.
    \label{eq:truncation_l2_bias}
\end{align}
We define the discarded target-weight fraction as $W_\mathrm{disc}:= \sum_{(\alpha,\beta)\notin\Omega_\tau}
    \chi_\mathcal{U}(\alpha,\beta)^2 /d^2$. Therefore, using that $\sum_{\alpha,\beta}\chi_\mathcal{E}(\alpha,\beta)^2\leq d^2$ gives 
\begin{equation}
    |F_\tau -F_e(\mathcal{U},\mathcal{E})|\leq \sqrt{W_\mathrm{disc}}.
\end{equation}
This makes $W_\mathrm{disc}$ a simple, computable diagnostic for the truncation bias. Since the target channel is known, $W_\mathrm{disc}$ can be evaluated directly as a function of the truncation threshold $\tau$,
\begin{equation}
    W_\mathrm{disc} (\tau)= \frac{\sum_{|\chi_\mathcal{U}(\alpha,\beta)|<\tau}
    \chi_\mathcal{U}(\alpha,\beta)^2}{d^2} \leq \frac{N_\mathrm{disc}\tau^2}{d^2} \leq d^2\tau^2,
\end{equation}
given that $N_\mathrm{disc}$ is defined as the number of discarded coefficients, which is bounded by $d^4$.

\section{Linear-inversion readout-error mitigation}
\label{app:linear_inversion_rem}

Here we recall the linear-inversion readout-error mitigation (REM) used in the
numerical comparisons. For a fixed group configuration $m$, let
$\mathbf p_m$ denote the ideal output distribution in the common measurement
basis and let $\mathbf q_m$ be the observed distribution. A calibrated assignment
matrix $A_m$ relates them by
\begin{equation}
    \mathbf q_m=A_m\mathbf p_m,
    \qquad
    \widehat{\mathbf p}_m^{\,\mathrm{REM}}
    =A_m^{-1}\widehat{\mathbf q}_m,
    \label{eq:linear_inversion_rem}
\end{equation}
provided $A_m$ is invertible. In the symmetric local-flip model used in the
figures, this matrix factorizes into one-qubit assignment matrices. If
$\boldsymbol\mu_{m,l}=(\mu^{(1)}_{m,l},\ldots,\mu^{(d)}_{m,l})^{\mathsf T}$ collects
the output-Pauli eigenvalues, the corrected estimate of the corresponding
expectation value is
\begin{equation}
    \widehat{\langle P_{\alpha_{m,l}}\rangle}_{\!\mathrm{REM}}
    =\boldsymbol\mu_{m,l}^{\mathsf T}A_m^{-1}
      \widehat{\mathbf q}_m .
    \label{eq:rem_corrected_expectation}
\end{equation}
Equivalently, at the single-shot level the raw parity in
Eq.~\eqref{eq:B_m_l_j} is replaced by
\begin{equation}
    \widetilde B_{m,l}^{(j)}
    =\lambda_{m,l}^{(b_j)}\boldsymbol\mu_{m,l}^{\mathsf T}
      A_m^{-1}\mathbf e_{a_j},
    \label{eq:rem_corrected_parity}
\end{equation}
where $\mathbf e_{a_j}$ is the one-hot vector for the observed output label.
The corrected variables are unbiased when $A_m$ is known exactly, but they need
not lie in $[-1,1]$; the inverse therefore removes readout bias at the cost of
an amplification of finite-shot fluctuations. A convenient worst-case
amplification factor for this Pauli is
$a_{m,l}=\lVert A_m^{-\mathsf T}\boldsymbol\mu_{m,l}\rVert_\infty$.
The grouped estimator is formed
as before by replacing $B_{m,l}^{(j)}$ with $\widetilde B_{m,l}^{(j)}$ in
$C_m^{(j)}$ and in the subsequent empirical aggregation of Eq.~\eqref{eq:hoeff}. Accordingly, the Hoeffding bound of Eq.~\eqref{eq:hoeff} and the allocation of Eq.~\eqref{eq:sm} continue to hold with $\lVert\mathbf{u}_m\rVert_1$ replaced by $\sum_l |u_{m,l}|a_{m,l}$.

Grouping also changes how this finite-shot cost is distributed. One sampled shot
from a grouped configuration contributes to the same empirical output
distribution (and its linear-inversion correction) for every compatible Pauli in
$G_m$. Under a fixed total
protocol-level circuit budget, the budget can consequently be concentrated on
fewer input--output settings than in the ungrouped protocol. This re-use can
amortize the sampling fluctuations introduced by $A_m^{-1}$ and explains the
observed per-configuration shot concentration. It does not make the assignment
matrix or its inverse intrinsically more accurate: the effect is conditional on
the grouping, shot-allocation rule, and equal-budget comparison, and can be
absent for flat or otherwise unfavorable groups.

\section{Repeated-cycle simulator noise model}
\label{app:simulator_noise_model}

Table~\ref{tab:repeated_cycle_noise_model} specifies the controlled simulator
model used for the fSim repeated-cycle comparison.
We construct a noise model in Qiskit Aer~\cite{qiskit_aer}
over the stated native gate basis. Gate-local thermal-relaxation channels are
attached directly to the physical $R_X$, $R_Y$, and $R_{ZZ}$
instructions; the two-qubit channel is the tensor product of the corresponding
single-qubit relaxation channels. Coherent residual unitaries are appended to
the same native instructions. The final readout model consists of two local
classical assignment channels, each attached to its physical one-qubit
\texttt{measure} instruction.

All protocol estimates are formed from sampled measurements. CAFE uses the
Qiskit Aer Sampler primitive directly, while DFE uses a customized Estimator primitive backed by that
same sampler to reconstruct Pauli expectations from sampled counts. Readout
error mitigation is disabled in both paths. Consequently, the assignment maps
contribute to the sampled estimates and to the fitted
$\epsilon_{\mathrm{SPAM}}$, but not to the exact forward-channel fidelity
used as the numerical reference. This model is a controlled
comparative test, not a fit to a calibrated hardware device.

\begin{table*}[t]
  \centering
  \caption{Qiskit Aer noise model for the fSim repeated-cycle comparison.
  For each readout matrix, rows label the ideal bit $0,1$ and columns the
  reported bit $0,1$.}
  \label{tab:repeated_cycle_noise_model}
  \begin{tabular}{@{}lll@{}}
    \toprule
    Component & Parameter & Value \\
    \midrule
    Simulator & Backend / native quantum basis & Qiskit Aer \texttt{AerSimulator} / $\{R_X, R_Y, R_{ZZ}\}$ \\
    Thermal relaxation & $T_1$ / $T_2$ & $100~\mu\mathrm{s}$ / $80~\mu\mathrm{s}$ \\
    Gate durations & $R_X$,$R_Y$ / $R_{ZZ}$& $40~\mathrm{ns}$ / $250~\mathrm{ns}$ \\
    Virtual gate & $R_Z$& No gate-local noise \\
    Coherent two-qubit residual & After each native $R_{ZZ}$& $R_{ZZ}(0.10)$ \\
    Coherent single-qubit residual & After each physical $R_X$,$R_Y$& $R_Z(0.01)$ \\
    Readout on qubit 0 & $P(\mathrm{reported}\mid\mathrm{ideal})$ & $\begin{psmallmatrix}0.996&0.004\\0.006&0.994\end{psmallmatrix}$ \\
    Readout on qubit 1 & $P(\mathrm{reported}\mid\mathrm{ideal})$ & $\begin{psmallmatrix}0.997&0.003\\0.005&0.995\end{psmallmatrix}$ \\
    Measurement estimator & CAFE / DFE & Aer Sampler / Custom Estimator built on same sampler \\
    Readout-error mitigation & Both protocols & Disabled \\
    \bottomrule
  \end{tabular}
\end{table*}

\section{Reinforcement-Learning Calibration Framework}
\label{app:reinforcement_learning_calibration}

To evaluate the practical impact of grouped Direct Fidelity Estimation (DFE) in a calibration setting, we employ a reinforcement-learning (RL) agent based on Proximal Policy Optimization (PPO), in the model-free quantum-control setting demonstrated by Sivak \emph{et al.}~\cite{Sivak_2022}. The calibration task can be viewed as a \emph{contextual bandit}: for a given target gate configuration (context), the agent proposes a set of continuous control parameters (actions), receives a fidelity-based reward, and updates its policy accordingly. No long-horizon state evolution is considered.

Table~\ref{tab:fsim_ppo_settings} collects the PPO and DFE settings used for the fSim learning-calibration results.  They are kept explicit because the action-batch size and DFE precision determine the protocol-level budget accumulated by each policy update.

\begin{table*}[t]
  \centering
  \caption{PPO and estimator settings for the fSim learning-calibration results in Fig.~\ref{fig:fsim_rl_convergence_comparison}.  The policy acts on the five parameters of Eq.~\eqref{eq:fsim_calibration_ansatz}; the $[-1,1]$ interval denotes the normalized action domain presented to PPO.}
  \label{tab:fsim_ppo_settings}
  \begin{tabular}{lll}
    \toprule
    Category & Parameter & Value \\
    \midrule
    Policy network & Hidden layers / activation & $[64,64]$ / tanh \\
    Policy distribution & Mean / standard-deviation output & tanh / sigmoid \\
    PPO update & Optimizer / learning rate & Adam / $5\times10^{-4}$ \\
    PPO update & Epochs / minibatch size & $8$ / $16$ \\
    PPO regularization & Clip ratio / entropy coefficient & $0.2$ / $0.01$ \\
    PPO objective & Value-loss coefficient / gradient clip & $0.5$ / $0.5$ \\
    Discounting & $\gamma$ / GAE $\lambda$ & $0.99$ / $0.95$ \\
    Calibration run & PPO updates / action batch & $30$ / $50$ \\
    Calibration run & Normalized action domain & $[-1,1]^5$ \\
    DFE reward & $(\epsilon,\delta)$ / shots per sampled setting & $(0.1,0.1)$ / $1$ \\
    \bottomrule
  \end{tabular}
\end{table*}

The central object of the learning procedure is the \emph{policy} $\pi_\theta(a)$, parameterized by trainable parameters $\theta$. In our implementation, the policy is represented as a multivariate Gaussian distribution over the gate-control parameter space. At each optimization step, the policy generates a batch of candidate control parameters,
$$
a_i \sim \pi_\theta(a), \qquad i=1,\ldots,B,
$$
with $B$ action samples per update. Each sampled action corresponds
to a distinct pulse or gate configuration that is evaluated on hardware (or
simulation) using the fidelity-estimation protocol described in this work.

The resulting fidelity estimates define rewards $r_i$, which are used to construct an empirical estimate of the expected policy performance,
$$
J(\theta)=\mathbb{E}_{a\sim\pi_\theta}[r(a)].
$$
The objective of learning is to maximize this expected reward, whose gradient admits the score-function form
\begin{equation}
\nabla_\theta J(\theta)
=
\mathbb{E}_{a\sim\pi_\theta}
\left[
\nabla_\theta \log \pi_\theta(a)\, r(a)
\right],
\end{equation}
which allows the policy parameters to be updated directly from sampled rewards. PPO introduces additional stabilization mechanisms that constrain successive policy updates, improving robustness and sample efficiency in practice.

The key advantage of this approach is that calibration remains entirely model-free: no analytical description of the device noise is required. The learning procedure relies solely on experimentally measured fidelity estimates and the ability to sample and optimize a differentiable policy over the control-parameter space.

\end{document}